\documentclass[11pt]{article}
\usepackage{makeidx}  % allows for indexgeneration
\usepackage{fullpage}
\usepackage{times}
\usepackage{float}
\usepackage{subfigure}
\usepackage{color}
\usepackage{url}
\usepackage{graphicx}
\usepackage{amsmath}
\usepackage{amssymb}

\usepackage{times}

\usepackage{tikz}
\usepackage{ifthen}
\usepackage{fp}
\usetikzlibrary{arrows,shapes}
\usepackage{algorithm}
\usepackage{algpseudocode}
\usetikzlibrary{decorations.pathmorphing}

\newcommand{\E}{\mathcal{E}}

\newcommand{\qed}{\mbox{}\hspace*{\fill}\nolinebreak\mbox{$\rule{0.6em}{0.6em}$}}

\newcommand{\expect}{{\bf \mbox{\bf E}}}
\newcommand{\prob}{{\bf \mbox{\bf Pr}}}

\definecolor{gray}{rgb}{0.5,0.5,0.5}
\newcommand{\e}{{\epsilon}}
\newcommand{\eps}{{\epsilon}}

\newcommand{\bool}{{\{0, 1\}}}
\newcommand{\G}{{\mathcal{G}}}
\newcommand{\YES}{{\bf YES~}}
\newcommand{\NO}{{\bf NO~}}
\newcommand{\ALG}{{\bf ALG~}}

\if 0
\newtheorem{theorem}{Theorem}[section]
\newtheorem{lemma}[theorem]{Lemma}
\newtheorem{claim}[theorem]{Claim}
\newtheorem{corollary}[theorem]{Corollary}
\newtheorem{definition}[theorem]{Definition}
\newtheorem{remark}[theorem]{Remark}

\newtheorem{observation}[theorem]{Observation}
\newtheorem{proposition}[theorem]{Proposition}
\fi
\newtheorem{lemma}{Lemma}[section]
\newtheorem{theorem}[lemma]{Theorem}

\newtheorem{claim}[lemma]{Claim}

\newtheorem{definition}[lemma]{Definition}

\newenvironment{proof}{{\bf Proof:}}{$\qed$\par}

\newenvironment{proofof}[1]{\noindent{\bf Proof of #1:}}{$\qed$\par}

\usepackage{hyperref}
\hypersetup{
bookmarksnumbered
}

{\makeatletter
 \gdef\xxxmark{%
   \expandafter\ifx\csname @mpargs\endcsname\relax % in minipage?
     \expandafter\ifx\csname @captype\endcsname\relax % in figure/caption?
       \marginpar{xxx}% not in a caption or minipage, can use marginpar
     \else
       xxx % notice trailing space
     \fi
   \else
     xxx % notice trailing space-
   \fi}
 \gdef\xxx{\@ifnextchar[\xxx@lab\xxx@nolab}
 \long\gdef\xxx@lab[#1]#2{{\bf [\xxxmark #2 ---{\sc #1}]}}
 \long\gdef\xxx@nolab#1{{\bf [\xxxmark #1]}}
}

\newcommand{\D}{{\mathcal D}}

\newcommand{\poly}{\text{poly}}

\newcommand{\lba}{n^{-1/10}}
\newcommand{\DBHP}{{\bf D-BHP~}}

\newcommand{\wh}[1]{\widehat{#1}}

\begin{document}
% \title{Approximating matching size in random graph streams}
\title{Streaming Lower Bounds for Approximating MAX-CUT}
\author{Michael Kapralov\thanks{IBM T. J. Watson Research Center, Yorktown Heights, NY 10598. Email: {\tt michael.kapralov@gmail.com} Work done while at MIT CSAIL. This research was supported by NSF award CCF-1065125, MADALGO center and Simons Foundation. We also acknowledge financial support from grant \#FA9550-12-1-0411
from the U.S. Air Force Office of Scientific Research (AFOSR) and the
Defense Advanced Research Projects Agency (DARPA).} \and Sanjeev Khanna\thanks{Department of Computer and Information Science, University of Pennsylvania,
Philadelphia, PA 19104. Email: {\tt sanjeev@cis.upenn.edu}.  Supported in part by National Science Foundation grant CCF-1116961.} \and Madhu Sudan\thanks{Microsoft Research New England, One Memorial Drive, Cambridge, 
MA 02142, USA. {\tt madhu@mit.edu}}}

\maketitle

\begin{abstract}
We consider the problem of estimating the value of max cut in a graph in the streaming model of computation. 
At one extreme, there is a trivial $2$-approximation for this problem that uses only $O(\log n)$ space, namely, count the number of edges and output half of this value as the estimate for max cut value. On the other extreme, if one allows $\tilde{O}(n)$ space, then a near-optimal solution to the max cut value can be obtained by storing an $\tilde{O}(n)$-size sparsifier that essentially preserves the max cut. An intriguing question is if poly-logarithmic space suffices to obtain a non-trivial approximation to the max-cut value (that is, beating the factor $2$). It was recently shown that the problem of estimating the size of a maximum matching in a graph admits a non-trivial approximation in poly-logarithmic space.

Our main result is that any streaming algorithm that breaks the $2$-approximation barrier requires $\tilde{\Omega}(\sqrt{n})$ space even if the edges of the input graph are presented in random order. Our result is obtained by exhibiting a distribution over graphs which are either bipartite or $\frac{1}{2}$-far from being bipartite, and establishing that $\tilde{\Omega}(\sqrt{n})$ space is necessary to differentiate between these two cases. Thus as a direct corollary we obtain that $\tilde{\Omega}(\sqrt{n})$ space is also necessary to test if a graph is bipartite or $\frac{1}{2}$-far from being bipartite. 
 We also show that for any $\epsilon > 0$, any streaming algorithm that obtains a $(1 + \epsilon)$-approximation to the max cut value when edges arrive in  adversarial order requires $n^{1 - O(\epsilon)}$ space, implying that $\Omega(n)$ space is necessary to obtain an arbitrarily good approximation to the max cut value.
 \end{abstract}
\setcounter{page}{0}

\thispagestyle{empty}
\newpage

%!TEX root = ./factor2.tex
\section{Introduction}

In the MAX-CUT problem an undirected graph is given as input, and the goal is to find a bipartition of the vertices of this graph (or, equivalently, a {\em cut}) that maximizes the number of edges that cross the bipartition. It is easy to find a solution to MAX-CUT that achieves a $2$-approximation: a uniformly random bipartition achieves this goal. The Goemans-Williamson algorithm~\cite{GW95} approximates MAX-CUT to a factor of $1.138$\footnote{The approximation ratio achieved by the Goemans-Williamson algorithm is usually stated as $0.878\cdots$ in the literature, but in this paper, we use the convention that approximation ratios are larger than $1$.} using semidefinite programming. This is best possible assuming the Unique Games Conjecture~\cite{KhotKMO07}.  In~\cite{Trevisan09} Trevisan presented an algorithm that achieves approximation ratio of $1.884$ using spectral techniques. A combinatorial algorithm that achieves approximation ratio strictly better than $2$ was presented by~\cite{KS11}. It is known that {\em dense graphs} are an easy case for this problem: polynomial time approximation schemes exist in graphs with $\Omega(n^2)$ edges~\cite{FK96, Vega96, AVKK03, AFNS09, MS08}. 

%%%It is also known that $(1+\e)$-approximation can be achieved via simple greedy methods in the dense case~\cite{MS08}.

All results mentioned above optimize the approximation ratio subject to polynomial (sometimes nearly linear) time complexity. However, in many settings {\em space complexity} of algorithms is a crucial parameter to optimize. For example, in applications to big data analysis one would like to design algorithms capable of processing large amounts of data using only few (ideally, a single) pass over the input stream and using limited (i.e. sublinear in input size) space. The {\em streaming model of computation}, formalized by Alon, Matias, and Szegedy~\cite{AlonMS96}, precisely captures this setting. Recently, the problem of developing streaming algorithms for fundamental graph problems has attracted a lot of attention in the literature (e.g. sparsifiers~\cite{anh-guha, kl11, agm_pods}, spanning trees~\cite{agm}, matchings~\cite{guha-ahn-1, guha-ahn-3, gkk:streaming-soda12, kap13, GO12, HRVZ13}, spanners~\cite{agm_pods, kw14}).  However, not much is known thus far on the space complexity of solving the MAX-CUT problem.

The goal of this paper is to understand how much space is necessary to obtain a good approximation to MAX-CUT value when the algorithm is given a single pass over a stream of edges of the input graph. The algorithms for MAX-CUT described above have natural streaming counterparts. For example, the trivial factor $2$ approximation algorithm that outputs a random bipartition leads to a simple factor $2$ approximation in $O(\log n)$ space: simply count the number of edges $m$ in the input graph and output $m/2$. If the input graph is dense, one can see that the techniques of~\cite{FK96, MS08} yield $(1+\e)$-approximation (for any $\eps > 0$) in $\poly(\log n)$ space in the streaming model using sampling. Finally, known results on sparsification in the streaming model~\cite{anh-guha, kl11, agm_pods} show that one can maintain a representation of the graph in $\tilde O(n/\e^2)$ space that preserves all cuts, and hence has sufficient information for obtaining a $(1+\e)$-approximate solution. This state-of-the-art, namely, a $2$-approximation in $O(\log n)$ space, and a $(1 + \e)$-approximation in $\tilde O(n/\e^2)$ space, highlight the following natural question: Can one approximate the max-cut value to a factor strictly better than $2$ in sub-polynomial space (say, poly-logarithmic)?  This is the precisely the question addressed in this work.

\iffalse

We resolve this in the negative by showing $\tilde{\Omega}(\sqrt{n})$ space is necessary for a streaming algorithm to achieve strictly better than a $2$-approximation. Our proof is based on showing that
any streaming algorithm that distinguishes between graphs that are bipartite and very far from being bipartite (require deleting almost half the edges to make them bipartite) with failure probability at most a constant $\delta \in [0,1/2)$, requires $\tilde{\Omega}(\sqrt{n})$ space.
Furthermore, our lower bound holds even if the edges graph are presented in the stream in a uniformly random order.  This is interesting, since algorithms for graph problems have recently been obtained that have better performance guarantees than the currently best known results in corresponding adversarial settings~\cite{kmm11, KKS14}.
\fi

%\paragraph{Property testing}
%Testing bipartiteness. $\Omega(\sqrt{n})$ lower bound known there, but it is a more restricted model. In streaming one accesses the whole %input, but has to store a small sketch at all times. Approximating vertex cover size~\cite{}.

\subsection{Our results}

Our main result is that $\tilde{\Omega}(\sqrt{n})$ space is necessary for a streaming algorithm to achieve strictly better than a $2$-approximation.

\begin{theorem}\label{thm:main}
Let $\e>0$ be a constant, and let $G=(V, E), |V|=n, |E|=m$ be an unweighted  (multi) graph. Any algorithm that, given a single pass over a stream of edges of $G$ presented in random order, outputs a $(2-\e)$-approximation to the value of the maximum cut in $G$ with probability at least $99/100$ (over its internal randomness) must use $\tilde \Omega(\sqrt{n})$ space.
\end{theorem}
Since a $2$-approximation can be obtained in $O(\log n)$ space by simply counting the edges,  our result rules out the possibility of any non-trivial approximation in sub-polynomial space, even for random streams. This makes progress on an open problem posed at the Bertinoro workshop on sublinear and streaming algorithms in 2011~\cite{ber11}.
The same conclusion carries over when the stream contains i.i.d. samples of the edge set of $G$.

\begin{theorem}\label{thm:main-iid}
Let $\e>0$ be a constant, and let $G=(V, E), |V|=n, |E|=m$ be an unweighted simple graph. Moreover, let $\ell$ be any positive integer less than  $\log^C n$ for some constant $C > 0$.
Any algorithm that, given  a single pass over a stream of $\ell \cdot n$ i.i.d. samples of $G$ presented in random order, outputs a $(2-\e)$-approximation to the value of the maximum cut in $G$ with probability at least $99/100$ (over its internal randomness) must use $\tilde \Omega(\sqrt{n})$ space.
\end{theorem}

Theorem~\ref{thm:main-iid} shows that it is hard to distinguish between random bipartite graphs and random non-bipartite graphs (with average degree about $1/\e^2$) presented as a stream of i.i.d.  samples of the edge set using substantially less than $\sqrt{n}$ space. We note that this result is tight up to polylogarithmic factors for our input distribution. A nearly matching algorithm is provided by the result of ~\cite{KKR04} for testing bipartiteness in graphs whose minimum and maximum degrees are within a constant factor of the average (algorithm \textsc{Test-Bipartite-Reg} in~\cite{KKR04}). Their algorithm performs $\tilde O(\sqrt{n})$ random walks of length $L=\poly(\log n)$ starting from a uniformly random node in $V$, and tests if the sets of vertices reached after an even number of steps intersects the set of vertices reached after an odd number of steps.  It is easy to see that this algorithm can be implemented in $\tilde O(\sqrt{n})$ space using a single pass over a stream of $\ell\cdot n$ i.i.d. samples of the input graph as long as $\ell\geq C' L \log n$ for a sufficiently large constant $C'>0$. Indeed, in order to run a randon walk it suffices to maintain the current vertex that the walk is at  and advance the walk one step as soon as an edge incident on the current node arrives in the stream. It takes $m/d=O(n)$ samples for the next edge incident on the current node to arrive, and hence $\ell\geq C'L\log n$ samples suffice to simulate a random walk of length $L$.

\iffalse
The space bound given by Theorem~\ref{thm:main-iid} is matches up to polylogarithmic factors the upper bound provided by~\cite{KKR04} for testing bipartiteness in graphs whose minimum and maximum degrees are within a constant factor of the average $d=O(\sqrt{n})$ (algorithm \textsc{Test-Bipartite-Reg} in~\cite{KKR04}). Their algorithm performs $\tilde O(\sqrt{n})$ random walks of length $L=\poly(\log n)$ starting from a uniformly random node in $V$, and tests if the sets of vertices reached after an even number of steps intersects the set of vertices reached after an odd number of steps.  It is easy to see that this algorithm can be implemented in $\tilde O(\sqrt{n})$ space using a single pass over a stream of $\ell\cdot n$ samples of the input graph as long as $\ell\geq C' L$ for a sufficiently large constant $C'>0$. Indeed, in order to run a randon walk it suffices to maintain the current vertex that the walk is at advance the walk one step as soon as an edge incident on the current node arrives in the stream. It takes $m/d=O(n)$ samples for the next edge incident on the current node to arrive, and hence $\ell\geq C'L$ samples suffice to run simulate a random walk of length $L$.
\fi

Finally, we also show that when the stream is adversarially ordered, any algorithm that can achieve an arbitrarily good approximation to the maxcut value, essentially requires linear space.

\begin{theorem}\label{thm:1pluseps}
For any $t\geq 2$ obtaining a $(1+1/(2t))$-approximation to the value of maxcut in the single pass adversarial streaming setting requires $\Omega(n^{1-1/t})$ space.
\end{theorem}

\paragraph{Recent related work.} Independently and concurrently, the 
task of finding lower bounds on the space required by streaming algorithms for finding approximate max-cuts in graphs was also 
explored by Kogan and Krauthgamer~\cite{kogan-krauthgamer14}. In particular, they also prove a theorem that is qualitatively similar to Theorem~\ref{thm:1pluseps}. Our proofs 
are also similar, though the exact gadget used in the reduction is somewhat different, leading to slightly different constants. Theorems~\ref{thm:main} and~\ref{thm:main-iid} in our 
work however seem new even given their results.

\subsection{Our techniques}

Our starting point is the lower bound for the so-called {\bf Boolean Hidden Matching (BHM)} problem due to Gavinsky et al.~\cite{GKKRW07} and its extension by Verbin and Yu~\cite{VY11}. {\bf BHM} is a two party one-way communication problem. Alice's input in {\bf BHM} is a boolean vector $x\in \bool^n$ and Bob's input is a matching $M$  of size $r=\Theta(n)$ on the set of coordinates $[n]$, as well as a vector $w\in \bool^r$. In the \YES case the vector $w$ satisfies $w=Mx$, where we identify the matching $M$ with its $r\times n$ edge incidence matrix, and in the \NO case $w=Mx\oplus 1^r$. In other words, in the \YES case endpoints of every edge $e=(u, v)\in M$ satisfy $x_u+x_v=w_{uv}$ and in the \NO case $x_u+x_v=w_{uv}+1$. Here arithmetic is modulo $2$. It was shown in ~\cite{GKKRW07}  that the randomized one-way communication complexity of {\bf BHM} is $\Omega(\sqrt{n})$. The first use of {\bf BHM} for streaming lower bounds was due to~\cite{VY11},  who also defined and proved lower bounds for a more general problem called {\bf Boolean Hidden Hypermatching} and used to to prove lower bounds for the streaming complexity of cycle counting, sorting by reversals and other problems.

Our Theorem~\ref{thm:1pluseps} is based on a simple reduction from the {\bf Boolean Hidden Hypermatching} problem of Verbin and Yu~\cite{VY11}. It shows that $(1+\e)$-approximation to maxcut value requires at least $n^{1-O(\e)}$ space when the stream is presented in adversarial order (section~\ref{sec:1pe}). This reduction is similar in spirit to the reduction from cycle counting presented in Verbin and Yu. The graph instances produced by the reduction contain about $\e n$ cycles of length about $1/\e$, and the length of these cycles is even in the \YES case and odd in the \NO case.  While this rules out a $(1+\e)$-approximations in small space, the approach appears to not be sufficiently robust for the proof of our main result in that (1) it heavily relies on the adversarial arrival order, and (2) it does not seem to extend to the factor $(2-\e)$-approximation, where we would need to rule out algorithms that distinguish between graphs that are essentially bipartite from graphs that are essentially as far from bipartite as possible.

We get around both complications by using the following approach. Our input graph instances are essentially random Erd\H{o}s-R\'{e}nyi graphs that are bipartite in the \YES case and non-bipartite in the \NO case. In order to achieve a factor $(2-\e)$-factor gap in maxcut value we choose the expected degree of a node to be $\Theta(1/\e^2)$ (section~\ref{sec:dist}).  The graphs are revealed to the algorithm in $\Omega(1/\e^2)$ {\em phases}, essentially corresponding to an $\Omega(1/\e^2)$-party one-way communication game. This allows us to ensure that graphs that arrive in each phase are subcriticial Erd\H{o}s-R\'{e}nyi graph, meaning that they are mostly unions of $O(\log n/\log\log n)$ size subtrees, and unlikely to contain cycles, and can thus convey only `local' information. While this distribution is natural, it is not immediately clear how to analyze it using techniques developed for the {\bf Boolean Hidden (Hyper)matching} problem. There are two issues here that we describe below. 

First, the {\bf BHM} problem is a two-party communication problem, while we are interested  in a $\Omega(1/\e^2)$-party communication game. However, we give a reduction from the {\bf BHM} problem (rather, a variation which we call the {\bf (Distributional) Boolean Hidden Partition} problem, or \DBHP; see below) to the  MAX-CUT problem on our instances. Roughly speaking, we show that any  algorithm  that solves MAX-CUT on our input instances must solve our two-party communication problem in at least one of the phases (see section~\ref{sec:reduction}). 
The second issue is that we would like to prove lower bounds that hold even for the setting where the input stream contains the  edges of the graph in a uniformly random order, but it is very unlikely that contiguous segments of a random stream of an edge set of a graph with average degree $\Omega(1/\e^2)$ form matchings. To remedy this, we introduce what we call the {\bf Boolean Hidden Partition}, or {\bf (D)-BHP} problem (see section~\ref{sec:comm-prob}).  In this problem Alice still gets a binary string $x\in \bool^r$ but Bob gets a general graph $G=(V, E), V=[n]$ together with parity information $w$ on the edges (thus, the special case when $G$ is a matching gives the {\bf BHM} problem of ~\cite{GKKRW07}). We show that this problem has a $\Omega(\sqrt{n})$ lower bound when $G$ is a subcritical Erd\H{o}s-R\'{e}nyi graph in section~\ref{sec:comm-prob}. These two ingredients already give a $\tilde\Omega(\sqrt{n})$ lower bound for streaming algorithms that achieve a factor $(2-\e)$-approximation to MAX-CUT in the adversarial order setting.
We then show that the arrival order of edges in our distribution is in fact close to uniformly random in total variation distance (with proper setting of parameters), yielding Theorem~\ref{thm:main}. Finally, we note that our reduction from MAX-CUT on instances that contain $k=\Theta(1/\e^2)$ `phases' turns out to be robust with respect to the number of phases $k$ -- the loss in terms of parameter $k$ is only polynomial. This allows us to also prove a lower bound for the setting where the input stream contains a sequence of $\ell\cdot n$ i.i.d. samples of the edge set of input graph for $\ell=\poly(\log n)$, yielding Theorem~\ref{thm:main-iid}.

\subsection{Organization} Section~\ref{sec:prelim} introduces some relevant concepts and notation. 
Section~\ref{sec:1pe} establishes Theorem~\ref{thm:1pluseps}. The rest of the paper is devoted to proving Theorem~\ref{thm:main} and Theorem~\ref{thm:main-iid}.
We define a hard input distribution for max cut in Section~\ref{sec:dist}. Then in Section~\ref{sec:comm-prob} we define the communication problem ({\bf Boolean Hidden Partition, BHP}), its distributional version \DBHP, and establish a $\tilde{\Omega}(\sqrt{n})$ lower bound for \DBHP. Section~\ref{sec:reduction} gives the reduction from  \DBHP to MAX-CUT. Theorem~\ref{thm:main} and Theorem~\ref{thm:main-iid} are then proved in Section~\ref{sec:main}.

%!TEX root = ./factor2.tex
\section{Preliminaries}\label{sec:prelim}

We will throughout follow the convention that $n$ denotes the number of vertices in the input graph $G$, and $m$ denotes the number of edges. We will use the notation $[n]=\{1,2,\ldots, n\}$. Also, for $x, y\in \bool$ we write $x+y$ or $x\oplus y$ denotes the sum of $x$ and $y$ modulo $2$.

\begin{definition}[Maxcut problem]
In the maxcut problem, we are given an unweighted graph $G=(V, E)$, and the goal is to output the value
$OPT:=\max_{P\cup Q=V, P\cap Q=\emptyset} |E\cap (P\times Q)|$,
that is, the maximum, over all bipartitions of $V$, of the number of edges of $G$ that cross the bipartition.
\end{definition}

Note that for any bipartite graph $G$, the maxcut value is $m$, and in general, the maxcut value of a graph is related to how far it is being from bipartite -- a notion formalized below.

\begin{definition}[$\beta$-far from bipartite]
For any $\beta \in [0,1/2]$, a graph $G=(V,E)$ is said to be $\beta$-far from being bipartite if any bipartite subgraph $G'$ of $G$ contains at most a $(1 - \beta)$-fraction of edges in $G$.
\end{definition}

If a graph $G$ is $\beta$-far from being bipartite, then maxcut value of $G$ is at most $(1 - \beta)m$.

\begin{definition}[$\gamma$-approximation to maxcut]
Let $G=(V, E)$ be a graph, and let OPT denote the maxcut value of $G$.
A randomized algorithm ALG is said to give a $\gamma$-approximation to maxcut with failure probability at most $\delta\in [0, 1/2)$ if on any input graph $G$, ALG outputs a value in the interval $[OPT/\gamma, OPT]$ with probability at least $1-\delta$.
\end{definition}

We will simply use the phrase $\gamma$-approximation algorithm for maxcut to refer to a $\gamma$-approximation algorithm with failure probability at most  $\delta = 1/4$.

Our focus will be on approximation algorithms for maxcut in the streaming model of computation where the edges of the graph are revealed to the algorithm in some order and the algorithm is constrained to use at most $c = c(n)$ space for some given space bound $c$. We will consider both the {\em adversarial} arrival model where the edges of the graph arrive in an order chosen by an oblivious adversary (i.e. adversary does not know any internal coin tosses of the algorithm) and the {\em random} arrival model where the edges of the graph arrive in a randomly permuted order (where the permutation is chosen uniformly at random). All our results concern single-pass streaming when the algorithm gets to see the edges of the graph exactly once.

Since the maxcut value is always bounded by $m$, and is always at least $m/2$ (take a uniformly random bipartition, for instance), there is a simple deterministic $2$-approximation streaming algorithm that uses $O(\log n)$ space: just count the number of edges $m$ and output $m/2$. On the other hand, for any $\eps > 0$, there is an $\tilde{O}(n/\eps^2)$ space streaming algorithm that computes a cut-sparsifier for a graph even in the adversarial arrival order. We can thus compute a $(1 + \e)$-approximation to max cut value in $\tilde{O}(n/\eps^2)$ space by first computing the sparsifier, and then outputting the maximum cut value in the sparsifier.

It is easy to see that any algorithm that computes a $\gamma$-approximation to maxcut value distinguishes between bipartite graphs and graphs that are $(1 - 1/\gamma)$-far from being bipartite. Thus in order to show that no streaming algorithm using space $c$ can achieve a $\gamma$-approximation with failure probability at most $\delta$, it suffices to show that no streaming algorithm using space $c$ can distinguish between bipartite graphs and graphs that are $(1 - 1/\gamma)$-far from being bipartite with probability at least $1 - \delta$. 

We conclude this section by defining the notion of total variation distance between probability distributions.
For a random variable $X$ taking values on a finite sample space $\Omega$ we let $p_X(\omega), \omega\in \Omega$ denote the pdf of $X$. For a subset $A\subseteq \Omega$ we use the notation $p_X(A):=\sum_{\omega\in A} p_X(\omega)$.
We will use the total variation distance $||\cdot ||_{tvd}$ between two distributions:
\begin{definition}[Total variation distance]
Let $X, Y$ be two random variables taking values on a finite domain $\Omega$.  We denote the pdfs of $X$ and $Y$ by $p_X$ and $p_Y$ respectively. The total variation distance between $X$ and $Y$ is given by
$V(X, Y)=\max_{\Omega'\subseteq \Omega} (p_X(\Omega')-p_Y(\Omega'))=\frac1{2}\sum_{\omega\in \Omega} |p_X(\omega)-p_Y(\omega)|$.
We will write $||X-Y||_{tvd}$ to denote the total variation distance between $X$ and $Y$.
\end{definition}

%!TEX root = ./factor2.tex
\section{An $n^{1-O(\e)}$ Lower Bound for $(1+\e)$-Approximation}\label{sec:1pe}

As a warm-up to our main result, we show here that for any $\e > 0$, a $(1+\e)$-approximation randomized streaming algorithm for max cut in the adversarial streaming model requires at least $n^{1-O(\e)}$ space. We will establish this result by a reduction from the {\bf Boolean Hidden Hypermatching} problem ({\bf BHH}) defined and studied by~\cite{VY11}.

\begin{definition}[$\text{BHH}_n^t$, Boolean Hidden Hypermatching]
The Boolean Hidden Hypermatching problem is a communication complexity problem where Alice gets a boolean vector
$x\in \bool^n$ where $n=2k t$ for some integer $k$, and Bob gets a perfect hypermatching $M$ on $n$ vertices where each edge contains $t$ vertices and a boolean vector $w$ of length $n/t$. Let $Mx$ denote the length $n/t$ boolean vector $(\bigoplus_{1\leq i\leq t} x_{M_{1, i}}, \ldots, \bigoplus_{1\leq i\leq t} x_{M_{n/t, i}})$ where $\{M_{1, 1,}, \ldots, M_{1, t}\}, \ldots, \{M_{n/t, 1}, \ldots, M_{n/t, t}\}$ are the edges of $M$. It is promised that either $Mx\oplus w=1^{n/t}$ or $Mx\oplus w=0^{n/t}$. 
The goal of the problem is for Bob to output
\YES when $Mx\oplus w=0^{n/t}$ and \NO when $Mx\oplus w=1^{n/t}$ ( $\oplus$ stands for addition modulo $2$).
\end{definition}

The following lower bound on the one-way communication complexity of $\text{BHH}_n^t$ was established in~\cite{VY11}.

\begin{theorem}{\rm \cite{VY11}}
\label{thm:bhh}
Any randomized one-way communication protocol for solving $\text{BHH}^t_n$ when $n=2kt$ for some integer $k\geq 1$  that succeeds with probability at least $3/4$ requires $\Omega(n^{1-1/t})$ communication.
\end{theorem}

We now give a proof of Theorem~\ref{thm:1pluseps}, which we restate here for convenience of the reader. The proof is via a reduction from $\text{BHH}^t_n$.

\noindent{\em {\bf Theorem~\ref{thm:1pluseps}}
For any $t\geq 2$ obtaining a $(1+1/(2t))$-approximation to the value of maxcut in the single pass adversarial streaming setting requires $\Omega(n^{1-1/t})$ space.
}

\noindent\begin{proof}
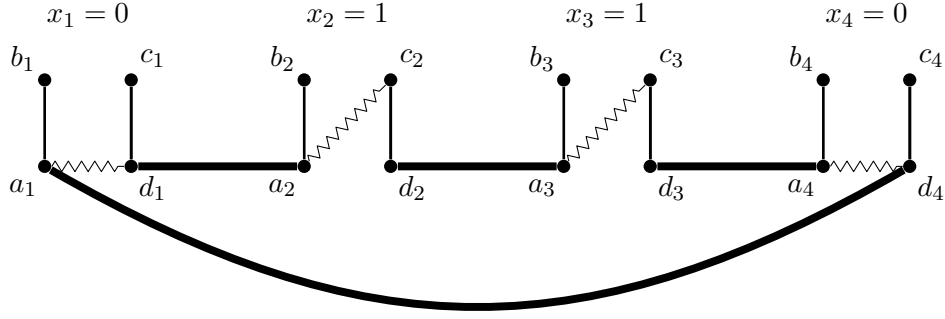
\begin{figure}[h!]
\begin{center}
\tikzstyle{vertex}=[circle, fill=black!100, minimum size=5,inner sep=1pt]
\tikzstyle{svertex}=[circle, fill=blue!30, minimum size=20,inner sep=1pt]
\tikzstyle{evertex}=[circle,draw=none, minimum size=25pt,inner sep=1pt]
\tikzstyle{edge} = [draw,-, color=red!100, very  thick]
\tikzstyle{bedge} = [draw,-, color=green!100, very  thick]
\begin{tikzpicture}[scale=1.15, auto,swap]

    \node at (0.5, 1.75) {$x_1=0$};

    \node at (-0.25, -0.25) {$a_1$};        
    \node at (-0.25, 1.25) {$b_1$};            
    \node at (1.25, 1.25) {$c_1$};                
    \node at (1.25, -0.25) {$d_1$};                
        
    \node[vertex](a1) at (0, 0) {};
    \node[vertex](b1) at (0, 1) {};    
    \node[vertex](c1) at (1, 1) {};
    \node[vertex](d1) at (1, 0) {};
    
    \path[draw, line width=1pt, -] (a1) -- (b1);
    \path[draw, line width=1pt, -] (c1) -- (d1);    

%   \draw [-,
%line join=round,
%decorate, decoration={
%    zigzag,
%    segment length=5,
%    amplitude=2.0,post=lineto,
%    post length=1pt
%}]  (a1) -- (c1);    
%    \path[draw, line width=1pt, -, dashed] (a1) -- (d1);    
   \draw [-,
line join=round,
decorate, decoration={
    zigzag,
    segment length=5,
    amplitude=2.0,post=lineto,
    post length=1pt
}]  (a1) -- (d1);

    \path[draw, line width=1pt, -] (c1) -- (d1);        

    \node at (3.5, 1.75) {$x_2=1$};

    \node at (2.75, -0.25) {$a_2$};        
    \node at (2.75, 1.25) {$b_2$};            
    \node at (4.25, 1.25) {$c_2$};                
    \node at (4.25, -0.25) {$d_2$};

    \node[vertex](a2) at (3, 0) {};
    \node[vertex](b2) at (3, 1) {};    
    \node[vertex](c2) at (4, 1) {};
    \node[vertex](d2) at (4, 0) {};
    
    \path[draw, line width=1pt, -] (a2) -- (b2);
    \path[draw, line width=1pt, -] (c2) -- (d2);    

   \draw [-,
line join=round,
decorate, decoration={
    zigzag,
    segment length=5,
    amplitude=2.0,post=lineto,
    post length=1pt
}]  (a2) -- (c2);    
%    \path[draw, line width=1pt, -, dashed] (a2) -- (d2); 

    \path[draw, line width=1pt, -] (c2) -- (d2);  
    
    \node at (6.5, 1.75) {$x_3=1$};

    \node at (5.75, -0.25) {$a_3$};        
    \node at (5.75, 1.25) {$b_3$};            
    \node at (7.25, 1.25) {$c_3$};                
    \node at (7.25, -0.25) {$d_3$};

    \node[vertex](a3) at (6, 0) {};
    \node[vertex](b3) at (6, 1) {};    
    \node[vertex](c3) at (7, 1) {};
    \node[vertex](d3) at (7, 0) {};
    
    \path[draw, line width=1pt, -] (a3) -- (b3);
    \path[draw, line width=1pt, -] (c3) -- (d3);    

   \draw [-,
line join=round,
decorate, decoration={
    zigzag,
    segment length=5,
    amplitude=2.0,post=lineto,
    post length=1pt
}]  (a3) -- (c3);    
%    \path[draw, line width=1pt, -, dashed] (a3) -- (d3); 

    \path[draw, line width=1pt, -] (c3) -- (d3);      

    \node at (9.5, 1.75) {$x_4=0$};

    \node at (8.75, -0.25) {$a_4$};
    \node at (8.75, 1.25) {$b_4$};
    \node at (10.25, 1.25) {$c_4$};
    \node at (10.25, -0.25) {$d_4$};
    
    \node[vertex](a4) at (9, 0) {};
    \node[vertex](b4) at (9, 1) {};    
    \node[vertex](c4) at (10, 1) {};
    \node[vertex](d4) at (10, 0) {};

    \path[draw, line width=1pt, -] (a4) -- (b4);
    \path[draw, line width=1pt, -] (c4) -- (d4);    

%   \draw [-,
%line join=round,
%decorate, decoration={
%    zigzag,
%    segment length=5,
%    amplitude=2.0,post=lineto,
%    post length=1pt
%}]  (a4) -- (c4);    
%    \path[draw, line width=1pt, -, dashed] (a4) -- (d4); 
   \draw [-,
line join=round,
decorate, decoration={
    zigzag,
    segment length=5,
    amplitude=2.0,post=lineto,
    post length=1pt
}]  (a4) -- (d4);

    \path[draw, line width=1pt, -] (c4) -- (d4);  
    
    %%%%Bob's edges
    \path[draw, line width=3pt, -] (d1) -- (a2);
    \path[draw, line width=3pt, -] (d2) -- (a3);    
    \path[draw, line width=3pt, -] (d3) -- (a4);        
    \draw[-, line width=3pt] (d4) to [out = -150, in = -30, looseness = 1.1    ] (a1);%, looseness = 2                
            
%    \path[draw, line width=3pt, <-, color=blue] (v1) -- (v4);        
%    \path[draw, line width=1pt, ->] (v1) -- (v5);

%    \path[draw, line width=1pt, ->] (v2) -- (v4);                
%    \path[draw, line width=1pt, ->] (v2) -- (v6);                

%    \path[draw, line width=1pt, ->] (v3) -- (v4);                
%    \path[draw, line width=3pt, <-, color=blue] (v3) -- (v5);                
%    \path[draw, line width=1pt, ->] (v3) -- (v6);                

%    \path[draw, line width=3pt, <-, color=blue] (s) -- node[above] {\scriptsize 1}(v1);                
%    \path[draw, line width=1pt, ->] (s) -- node[above] {\scriptsize 1}(v2);                
%    \path[draw, line width=3pt, <-, color=blue] (s) -- node {\scriptsize 1}(v3);                
    
%    \path[draw, line width=3pt, <-, color=blue] (v4) -- node[above] {\scriptsize 1}(t);                
%    \path[draw, line width=3pt, <-, color=blue] (v5) -- node[above] {\scriptsize 1}(t);                
%    \path[draw, line width=1pt, ->] (v6) -- node {\scriptsize 1}(t);                
    
\end{tikzpicture}
\end{center}
\caption{Reduction from Boolean Hidden Hypermatching to approximating max cut value}\label{fig:bhh2maxcut}
\end{figure}

Let \ALG be a streaming algorithm that achieves a $(1+\e)$-approximation to the value of maxcut in the adversarial streaming model using space $c$. We will show that \ALG can be used to obtain a protocol for $\text{BHH}^t_n$ with one-way communication complexity of $c$. The space lower bound then follows from Theorem~\ref{thm:bhh}.

Let $x\in \bool^{n}, n=2kt$ denote the vector that Alice receives. Alice creates her part of the graph that will be given as input to \ALG as follows.
For each $i\in [n]$ create four vertices $a_i, b_i, c_i, d_i$ and add the following edges to the set $E_1$ (see Fig.~\ref{fig:bhh2maxcut}). If 
$x_i=0$, add edges $(a_i, b_i), (c_i, d_i)$ and the edge $(a_i, d_i)$. Otherwise add edges $(a_i, b_i), (c_i, d_i)$ and the edge $(a_i, c_i)$.
Alice then treats $E_1$ as the first half of the stream, runs \ALG on $E_1$ and sends the state of \ALG to Bob.

Bob constructs a set of edges $E_2$ as follows. For each pair $(M_i, w_i)$ that Bob receives he creates $t$ edges as follows. Bob adds the following sets of edges for each hypermatching $M_i$, $i\in [2k]$, depending on $w_i$ (denote the vertices in $M_i$ by $\{j_{1},j_{2},\ldots, j_{t}\}, j_{s-1}\leq j_{s}$ for all $s=2,\ldots, k$). If 
$w_i=0$, add edges $(d_{j_{s-1}}, a_{j_{s}})$ for $j=2,\ldots, t$ and the edge $(a_{j_{t}}, d_{j_{1}})$. Otherwise 
add edges $(d_{j_{s-1}}, a_{j_{s}})$ for $j=2,\ldots, t$ and the edge $(b_{j_{t}}, d_{j_{1}})$.

Bob treats $E_2$ as the second half of the stream, and completes the execution of \ALG on the stream, starting from the state of \ALG that was communicated by Alice. Let $m=|E_1\cup E_2|=(n/t)\cdot (4t)=4n$. If \ALG reports that max-cut is strictly larger than $(1-1/(4t))m$, Bob outputs \YES, otherwise \NO.

We now prove correctness. First note that the graph $E_1\cup E_2$ contains exactly $n/t$ cycles. These cycles can be indexed by hyperedges $M_i$ that Bob received, and they are edge disjoint. Note that the number of edges on the cycle corresponding to hyperedge $M_i$ is equal to $2t+w_i+\sum_{s=1}^t x_{j_s}$, where $M_i=\{j_1,j_2,\ldots, j_t\}$. Thus, the length of the cycle is even iff $\sum_{s=1}^t x_{j_s}=w_i$. Thus, if the $\text{BHH}_n^t$ instance is a \YES instance, the graph $E_1\cup E_2$ is bipartite, and the graph $E_1\cup E_2$ contains $n/t$ edge disjoint cycles otherwise. In the former case the maxcut value is $m$. In the latter case any bipartition will be avoided by at least one edge out of the $n/t$ odd cycles. Thus, the maxcut value is at most $m-n/t\leq (1-1/(4t))m$. Since $\frac{m}{(1-1/(4t))m}\leq 1+1/(2t)$ for all $t\geq 2$, this completes the proof of correctness of the reduction. 
\end{proof}

%!TEX root = ./factor2.tex
\section{Hard Input Distribution}\label{sec:dist}
The essence of the hard instances of the previous section were (hidden) odd cycles. The \YES  instances were roughly unions of cycles of length $2t$ while \NO instances were unions of cycles of length $2t+1$. The gap between the maxcut value in the two cases is a factor of roughly $1+ 1/(2t+1)$ and $n^{1 - \Theta(1/t)}$ space is necessary and sufficient for distinguishing these cases. To go further and establish a factor of $2-\epsilon$ hardness, we need a new class of \YES and \NO instances (and distributions supported on these) with the following features. The maxcut value between the two classes should be separated by a factor of $2- \epsilon$. At the same time the \NO instances should not contain small odd cycles (since these can potentially be detected with sublinear space. The class of distributions we consider are (minor variants of) random graphs with linear edge density for the \NO instances, and random bipartite graphs with the same edge density for the \YES instances. It is clear that these two distributions satisfy the properties we seek. Proving streaming lower bounds however is not immediate and in this section we give variants of the above distributions for which we are (in later sections) able to prove $\tilde{\Omega}(\sqrt{n})$ lower bounds on the space complexity of streaming algorithms that distinguish the two.

The basic hard distribution that we will work with is defined in section~\ref{sec:input-dist} and denoted by $\D$. This distribution is a uniform mixture of two distributions: $\D^Y$ (the \YES case distribution) is supported on bipartite random graphs with $\Theta(n/\e^2)$ edges, and $\D^N$ (the \NO case distribution) is supported on random non-bipartite graphs of the same density. (For technical reasons we allow our graphs to be multigraphs, i.e., with multiple edges between two vertices.) The density of our input graphs is crucial for obtaining a $2-\e$ gap. Our input instances are natually viewed as consisting of $k$ phases with $k=\Omega(1/\e^2)$, where during each phase a sparse (or, more precisely, subcritical) random graph is presented to the algorithm. Since the graph is sparse, the algorithm only obtains local information about its structure in each phase. In particular, the graph presented in each round is very likely to be a union of trees of size $O(\frac{\log n}{\log\log n})$. In order to ensure that graphs that appear in individual phases do not contain cycles (i.e. global information), we introduce a parameter $\alpha$ that controls the expected number of edges arriving in each phase. Thus, $\Theta(\alpha n)$ edges of $G$ arrive in each phase in expectation, and we have $k=\Theta(1/(\e^2 \alpha))$ phases. The number of phases $k$ is chosen as $\Theta(1/(\alpha \e^2))$, where $\alpha>0$ satisifies $\alpha>\lba$.

In what follows we first define the Erd\H{o}s-R\'{e}nyi family of random graphs and then define the distribution $\D$.

\subsection{Erd\H{o}s-R\'{e}nyi graphs}\label{sec:rgraphs}

Our input distribution will use Erd\H{o}s-R\'{e}nyi graphs, which we will denote by $\G_{n, p}$. Sampling a graph $G=(V, E)$ from the distribution $\G_{n, p}$ amounts to including every potential edge $\{i, j\}\in {V \choose 2}$ into $E$ independently with probability $p$. Our input distribution will be naturally viewed as consisting of $\Theta(1/(\alpha \e^2))$ {\em phases}. During each phase the graph arriving in the stream will (essentially) be drawn from $\G_{n, \alpha/n}$. Here $\alpha<1$ is a parameter that we will set later. Since $\alpha<1$, our graphs are {\em subcritical}. In particular, they are composed of small connected components (of size $O(\log n/\log\log n)$) with high probability. We will need several structural properties of graphs sampled from $\G_{n, \alpha/n}$, which we now describe.

\begin{definition}[Complex and unicyclic connected component]
Let $G=(V, E)$ be a graph, and let $C\subseteq V$ be a connected component of $G$. The component $C$ is called {\em complex} if the number of edges induced by $C$ is strictly larger than $|C|$, i.e. $|E\cap (C\times C)|>|C|$. The component $C$ is called unicyclic if it induces exactly $|C|$ edges, i.e. when the induced subgraph is connected and has exactly one cycle.
\end{definition}

We will use the fact that complex components are rare in graphs drawn from $\G_{n, \alpha/n}$:
\begin{lemma}[Lemma 2.6.1 in~\cite{durrett}]\label{lm:complex-comp}
The probability that $G=(V, E)$ sampled from $\G_{n, \alpha/n}$ for $\alpha<1$ contains a complex connected component is bounded by $O(\frac1{n}\alpha^2\log^4 n)$.
\end{lemma}

Unicyclic components are more frequent than complex components, but still quite rare, as the following lemma shows. We will need to choose the parameter $\alpha$ appropriately to avoid unicyclic components, i.e. that the graphs presented to the algorithm in each phase do not contain cycles.
\begin{lemma}\label{lm:cycles}
Let $G=(V, E)$ be sampled from $\G_{n, \alpha/n}$ for some $\alpha\in (\lba, 1)$. Then the probability that $G$ contains a cycle is bounded by $O(\alpha^3)$. 
\end{lemma}
\begin{proof}
The number of unicyclic graphs on $k$ vertices is bounded by~\cite{durrett} (page~54, above eq. (2.6.6))
\begin{equation}\label{eq:nu-k}
\nu_k=\frac{(k-1)!}{2}\sum_{j=0}^{k-3} \frac{k^j}{j!}\leq \frac{e^k (k-1)!}{2}
\end{equation}
Thus, the expected number of unicyclic components of size $k$ is bounded by~\cite{durrett} (page~54, eq. (2.6.7))
\begin{equation*}
\begin{split}
{n \choose k}\nu_k \left(\frac{\alpha}{n}\right)^k \left(1-\frac{\alpha}{n}\right)^{k(n-k)+{k \choose 2}-k}.
\end{split}
\end{equation*}
Summing this expression over all $k$ and using \eqref{eq:nu-k}, we get 
\begin{equation*}
\begin{split}
\sum_{k=3}^{+\infty} {n \choose k}\nu_k \left(\frac{\alpha}{n}\right)^k \left(1-\frac{\alpha}{n}\right)^{k(n-k)+{k \choose 2}-k}&\leq \sum_{k=3}^{+\infty} {n \choose k}\frac{e^k (k-1)!}{2} \left(\frac{\alpha}{n}\right)^k \left(1-\frac{\alpha}{n}\right)^{k(n-k)+{k \choose 2}-k}\\
&\leq \sum_{k=3}^{+\infty} {n \choose k}\frac{e^k (k-1)!}{2} \left(\frac{\alpha}{n}\right)^k \exp\left(-\alpha k(1-k/n)-\alpha {k \choose 2}/n+\alpha k/n\right)\\
&\leq \sum_{k=3}^{+\infty} {n \choose k}\frac{e^k (k-1)!}{2} \left(\frac{\alpha}{n}\right)^k,\\
\end{split}
\end{equation*}
where we used the fact that $k\geq 3$ to go from second to last line to the last line.

We now bound ${n\choose k}\leq n^k/k!$ to get 
\begin{equation*}
\begin{split}
&\sum_{k=3}^{+\infty} {n \choose k}\frac{e^k (k-1)!}{2} \left(\frac{\alpha}{n}\right)^k\leq \sum_{k=3}^{+\infty} \frac{e^k}{2} \alpha^k \leq \sum_{k=3}^{+\infty} \frac{e^{k+1}}{2} \alpha^k =O(\alpha^3)
\end{split}
\end{equation*}
whenever $\alpha$ is smaller than an appropriate constant.
\end{proof}

\subsection{Input distribution}\label{sec:input-dist}
We now define a distribution over input instances. The distribution, which we denote by $\D$, is a uniform mixture of two distributions: the \YES case distribution $\D^Y$ and the \NO case distribution $\D^N$. Thus, $\D=\frac1{2}\D^Y+\frac1{2}\D^N$. Graphs drawn $\D^Y$ will be bipartite, while graphs drawn from $\D^N$ will be almost $\frac1{2}$-far from bipartite. In other words, graphs  drawn from $\D^Y$ have maxcut value $m$, while graphs drawn from $\D^N$ have maxcut value at most $(1/2+\e)m$. Thus, showing that no $o(\sqrt{n})$ space algorithm can distinguish between $\D^Y$ and $\D^N$ will be sufficient to rule out $(2-O(\e))$-approximation to maxcut value in $o(\sqrt{n})$ space. 

In order to ensure a factor $2-\e$ gap between maxcut values in $\D^Y$ and $\D^N$, we make our input graph $G'=(V, E')$ a union of  $k=\Theta(\frac1{\alpha \e^2})$ sparse Erd\H{o}s-R\'{e}nyi graphs that we now define. Let $R=(P, Q)$ be a bipartition of $V$ generated by choosing a string $x\in \{ 0,1 \}^n$ uniformly at random and assigning every vertex $u\in V$ with $x_u=0$ to $P$ and  every vertex $u$ with $x_u=1$ to $Q$. The distribution of \YES-instances and \NO-instances is created as follows. First for each $i=1,\ldots, k$ sample $G_i=(V, E_i)\sim \G_{n, \alpha/n}$. Then

\begin{description}
\item[YES]  Generate $R=(P, Q)$ uniformly at random.  Let $G'_i=(V, E'_i)$ be the graph obtained by including those edges in $E_i$ that cross the bipartition $R$ (i.e. $E'_i\subseteq P\times Q$). Let $E':=E'_1\cup E'_2\cup \ldots \cup E'_{k}$.
\item[NO]   Let $G'_i=(V, E'_i)$ be the graph obtained by including each edge in $E_i$ independently with probability $1/2$. Let $E':=E'_1\cup E'_2\cup \ldots \cup E'_{k}$.
\end{description}

We denote the input distribution defined above by $\D^{Y}$ (\YES case) and $\D^{N}$ (\NO case) respectively. Let $\D=\frac1{2}\D^{Y}+\frac1{2}\D^{N}$. We note that the graphs generated by our distribution $\D$ are in general multigraphs. The expected number of repeated edges is only $O(1/\e^2)$, however.

We show that in the \YES case the value of maxcut is equal to all edges of the graph, and in the \NO case the maxcut is close to $m/2$: 
\begin{lemma}\label{cl:gap}
Let $G=(V, E), |V|=n, |E|=m$ be generated according to the process above, where $k=C/(\alpha\e^2)$ for a sufficiently large constant $C>0$. Then in the \YES case the maxcut is $m$, and in the \NO case the maxcut is at most $(1+\e)m/2$ whp.
\end{lemma}
The proof uses the following version of Chernoff bounds.

\begin{theorem}[\cite{book-random-graphs}, Theorems 2.1 and 2.8]\label{thm:chernoff}
Let $X=\sum_{i=1}^n X_i$, where $X_i$ are independent Bernoulli 0/1 random variables with expectation $p_i$. Let $\mu=\sum_{i=1}^n p_i$. Then for all $\Delta>0$
$$
\prob[X\geq \mu+\Delta]\leq \exp\left(-\frac{\Delta^2}{2\mu+2\Delta}\right).
$$
and
$$
\prob[X\leq \mu-\Delta]\leq \exp\left(-\frac{\Delta^2}{2\mu}\right).
$$

\end{theorem}

\noindent\begin{proofof}{Lemma~\ref{cl:gap}}
In the \YES case, all edges in $E$ go across the bipartition $R$, so the maxcut has size $m$. A straightforward application of Chernoff bounds shows that the with high probability, the value of $m$ is at least $(1 - O(\sqrt{\log n/n}))\frac{\alpha k n}{4}$ with high probability.

We now consider the \NO case. Fix a cut $(S, \bar{S})$ where $S \subseteq V$ and $|S|\leq n/2$. Let $r:=|S|$. The expected number of edges (counting multiplicities) that cross the cut $S$ is given by $r\cdot (r-|S|) (k\alpha)/(2n)$, which is maximized when $r=n/2$. The maximum is equal to $(\alpha k)n/8$. 
The probability that the actual value of the cut exceeds $(1+\e)m/2\geq (1+\e/2)(\alpha k)n/8$  is bounded as
$$
\prob[|E\cap S\times \bar S|>t]\leq \exp\left(-\frac{(t-\mu)^2}{2t}\right),
$$
where we let $t=\mu+\Delta=(1+\e/2)(\alpha k)n/8$ and $\mu=r\cdot (r-|S|) k(\alpha)/(2n)$. The right hand side is minimized when $\mu$ is maximized, which corresponds to $r=n/2$. Thus, the maximum expected cut size over all cuts equals $(\alpha k)n/8$, and hence we let $\mu:=(\alpha k)n/8$, $t=(1+\e/2)\mu$ and conclude that for any cut $S$
$$
\prob[|E\cap S\times \bar S|>t]\leq \exp\left(-\frac{\e^2 \mu}{8}\right)=\exp\left(-\frac{\e^2 (\alpha k)n/8}{8}\right).
$$
Now using the assumption that $k=C/(\alpha \e^2)$, we get
\begin{equation*}
\begin{split}
\prob[|E\cap S\times \bar S|>(1+\e)m/2]&\leq\prob[|E\cap S\times \bar S|>(1+\e/2) (\alpha k)n/8]\\
&\leq \exp\left(-\frac{\e^2 (\alpha k)n/8}{8}\right)\leq \exp\left(-Cn/64\right)<2^{-2n}
\end{split}
\end{equation*}
as long as $C>0$ is larger than an absolute constant. A union bound over at most $2^n$ cuts completes the proof.
\end{proofof}

Note that each graph $G_i$ considered in the distributions $\D^Y$ and $\D^N$ is a simple graph. Thus we also have that each $G'_i$ is a simple graph. However, the graph $G'$ being the union of simple graphs need not be simple. Indeed $G'$ will contain multiple edges with rather high probability. In later sections we will argue that a streaming algorithm with limited space will not be able to distinguish $\D^Y$ from $\D^N$, given access to edges of $E'_1$ in random order, and then edges of $E'_2$ in random order and so on.  We will then claim that this also applies to streaming algorithms that are given edges of $G'$ in a random order. We note here that these two input orderings are not the same: In particular, a random ordering of edges of $G'$ might include two copies of a multi-edge within the first $\alpha n$ edges, while $E'_1$ does not contain two such edges. In what follows (in Lemma~\ref{lm:perm}) we show that despite this difference between the random ordering and the ``canonical random ordering'' (alluded to above, and to be defined next), the two orderings are close in total variation distance and allowing us to reason about the latter to make conclusions about the former.

The edges appear in the stream in the order $E_1',E_2',\ldots, E_k'$, and order of arrival in each group $E_i', i=1,\ldots, k$ is uniformly random. We refer to this ordering as the {\em canonical random ordering}:
\begin{definition}\label{def:canonical}
Let $G'=(V, E'), E'=E_1'\cup\ldots\cup E'_k$ denote the set of edges generated by the process above. We refer to the ordering of the edges $E'$ given by $E_1',E_2',\ldots, E_k'$, where edges inside each $E'_i$ are ordered uniformly at random as the {\em canonical random ordering} associated with $\D$.
\end{definition} 

\begin{lemma}\label{lm:perm}
Let $\e>0$ be a constant, and $\alpha\in (\lba, 1)$ a parameter. Let $k=\Theta(1/(\alpha \e^2))$ and let $E_1',\ldots, E_k'$ denote the edges sets of graphs $G_1',\ldots, G_k'$ drawn from distribution $\D$, and let $G'$ be the union of these graphs.  Then the canonical random ordering of edges of $G'$ is $O(\alpha \log (1/\alpha))$-close to uniformly random in total variation with probability at least $1-o(\alpha)$ over the choice of randomness used to sample $G'$.
\end{lemma}
\begin{proof}
Let $E'$ be the set of edges of $G'$ sampled from the distribution $\D$, and let $\Pi$ denote the canonical random ordering (see Definition~\ref{def:canonical}; note that $\Pi$ is a random variable). 

 Consider an ordering $\pi$ of the edge set $E'$ (recall that in general $E'$ is a multiset). 
  Suppose that there exists an edge $e\in E'$ that is included in $E'$ at least twice (let two copies of $e$ be denoted by $e^1$ and $e^2$) such that  $|\pi(e^1)-\pi(e^2)|\leq 4\alpha n$, i.e. $e^1$ and $e^2$ arrive at distance at most $4\alpha n$ in the permutation $\pi$. Then we refer to $\pi$ as a {\em collision inducing} permutation.  In what follows we show that {\bf (1a)} the canonical random ordering $\Pi$ produces every  {\em non-collision-inducing} permutation with equal probability, and {\bf (1a)} produces any other permutation with only smaller probability.  We then show that {\bf (2)} collision inducing permutations are quite unlikely in the uniformly random ordering, which gives the result.
  
Consider the process of sampling from the distribution $\D$ that generated the set $E'$, and let $E'=E'_1\cup E'_2\cup \ldots\cup E'_k$ denote the sets that each of the $k$ phases of our generation process produced. Define the event $\E=\{|E'_i|\leq \alpha n\text{~for all~}i=1,\ldots,k\wedge |E'|\geq n\}$, i.e. the event that none of the individual sets $E'_i$ are too much larger than their expected size (which is about $\alpha n/4$), and that the set $E'$ itself is not too much smaller than its expected size.
Since $\expect[|E'_i|]\leq \alpha n/4$ in both \YES and \NO cases, and that $\expect[|E'|]>2n$ for sufficiently small constant $\e>0$, we have $\prob[\E]>1-e^{-\Omega(\alpha n)}$.  We condition on $\E$ in what follows.

We now give a proof of {\bf (1)}. Consider two permutations $\pi, \pi'$ that are not collision inducing, so that no two copies of an edge are at distance at most $4\alpha n$ under $\pi, \pi'$. Note that by conditioning on $\E$ this means that both $\pi$ and $\pi'$ are generated by $\Pi$ conditional on $\E$ with nonzero probability, since none of  $E'\cap E'_i, i=1,\ldots, k$ would need to contain duplicate edges. We now show that in fact $\pi$ and $\pi'$ are generated with equal probability by $\Pi$. Note that both in the \YES and \NO cases the distributions that the graphs $G'_i$ are drawn from are symmetric in the sense that the probability of a graph $G'_i$ generated only depends on the number of edges in the graph as long as the graph does not have repeated edges, and is zero otherwise. The latter case is excluded by the assumption that $\pi$ and $\pi'$ are not collision inducing, and hence
\begin{equation*}
\begin{split}
\prob[\Pi=\pi |E', \E]&=\expect_{E'=E'_1\cup\ldots\cup E'_k}\left[\prob[\Pi=\pi| E'_1,\ldots, E'_k, \E]\right]\\
&=\expect_{E'=E'_1\cup\ldots\cup E'_k}\left[\prob[\Pi=\pi' | E'_1,\ldots, E'_k, \E]\right]\\
&=\prob[\Pi=\pi' |E', \E].
\end{split}
\end{equation*}
This establishes {\bf (1a)}.

The same reasoning shows that if $\pi$ is collision inducing and $\pi'$ is not, then $\prob[\Pi=\pi |E', \E]\leq \prob[\Pi=\pi' |E', \E]$.
Indeed, for any $E'_1\cup \ldots E'_k=E'$ we have $\prob[\Pi=\pi | E'_1,\ldots, E'_k, \E]\leq \prob[\Pi=\pi' | E'_1,\ldots, E'_k, \E]$. This is because if neither $\pi$ not $\pi'$ map two copies of some edge to a single set $E'_i$ then the two terms are equal. If $\pi$ maps two copies of an edge to the same set $E'_i$, then the left term is zero. Thus, we have

\begin{equation*}
\begin{split}
\prob[\Pi=\pi |E', \E]&=\expect_{E'=E'_1\cup\ldots\cup E'_k}\left[\prob[\Pi=\pi| E'_1,\ldots, E'_k, \E]\right]\\
&\leq \expect_{E'=E'_1\cup\ldots\cup E'_k}\left[\prob[\Pi=\pi' | E'_1,\ldots, E'_k, \E]\right]\\
&=\prob[\Pi=\pi' |E', \E],
\end{split}
\end{equation*}
establishing {\bf (1b)}.

We now bound the number of collision inducing permutations for a typical set $E'$.  First let $\E^*$ denote the event that $E'$ contains no edges of multiplicity more than $2$ and $O(\log(1/\alpha))$ edges of multiplicity $2$. We now prove that $\prob[\E^*]=1-o(\alpha)$.

We start by bounding the expected number of edges of multiplicity $3$ or above.  The set $E'$ is a union of $k=\Theta(1/(\alpha \e^2))$ Erd\H{o}s-R\'{e}nyi graphs, so 
a union bound over all ${n \choose 2}$ potential edges and ${k \choose 3}$ potential phases that $3$ copies of the edge should appear in shows that the expected number of edges with multiplicity $3$ and above is bounded by 
\begin{equation}\label{eq:123fwef}
{n \choose 2} (\alpha/n)^3 {k \choose 3}\leq O(1/(\e^3n))=O(1/n)
\end{equation}
since $\e$ is a constant. Thus, there are no such edges with probability at least $1-o(\alpha)$ (using the assumption that $\alpha\leq \lba$).

We now bound the number of edges of multiplicity $2$.  The number of such edges is a sum of ${n \choose 2}$ Bernoulli random variables with expectations bounded by ${k\choose 2} (\alpha/n)^2=O(\frac1{(\alpha \e^2)^2})\cdot (\alpha/n)^2=O(1/(\e^2n)^2)$. Thus, the expected number of duplicate edges is $O(1/\e^4)=O(1)$ by the assumption that $\e$ is an absolute constant, and this number is bounded by $O(\log (1/\alpha))$ with probability at least $1-o(\alpha)$ by standard concentration inequalities.  Putting this together with \eqref{eq:123fwef}, we get
that  
$$
\prob[\E^*]\geq 1-O(1/(\alpha n))-o(\alpha)=1-o(\alpha),
$$
where we used the fact that $\alpha>\lba$ by assumption. 

We now bound the number of collision inducing permutations $\pi$ conditional on $\E^*$. The probabilty that a uniformly random $\pi$ maps two copies of an edge within distance $4\alpha n$ is bounded by $4\alpha n/|E'|=O(\alpha)$. By a union bound over $O(\log(1/\alpha))$ edges of multiplicity $2$ the fraction of collision inducing permutations is $O(\alpha \log(1/\alpha))$ as required.

We have shown that permutations that are not collision inducing are equiprobable, and at least as probable collision inducing permutations, which amount to an $O(\alpha\log (1/\alpha))$ fraction of all permutations conditional on an event $\E^*\wedge \E$ that occurs with probability at least $1-o(\alpha)$. %Thus, we have that for any $\pi$ that is not collision inducing 
%$$
%(1-O(\alpha \log(1/\alpha)))/|E'|!\leq \prob[\Pi=\pi |E', \E]\leq 1/|E'|!
%$$
Thus, the total variation distance between the uniformly random ordering 
and the canonical random ordering is $O(\alpha \log(1/\alpha))$ with probability at least $1-o(\alpha)$ over the randomness used to sample $G'$.
\end{proof}

%!TEX root = ./factor2.tex
\section{The Boolean Hidden Partition Problem}\label{sec:comm-prob}
We analyze the following $2$-player one-way communication problem. 

\paragraph{Boolean Hidden Partition Problem (BHP).} Alice gets a vector $x\in \bool^n$. Bob gets a the edges of a graph $G=(V, E), V=[n], E\subseteq {[n] \choose 2}$, and a vector $w\in \bool^r$, where  $r$ denotes the number of edges in $G$. Note that we associate edges of $G$ with $[r]$. Let $M\in \bool^{r\times n}$ denote the edge incidence matrix of $G$, i.e. for each $e\in E$ and $v\in V$ $M_{ev}=1$  iff $v$ is an endpoint of $e$. 
Then (1) in the \YES case the vector $w$ satisfies $Mx=w$ (arithmetic is over $\mathbb{GF}(2)$); and (2) in the \NO case $w$ is uniformly random in $\bool^r$ independent of $x$.\footnote{Note that this is somewhat different from the setting of~\cite{GKKRW07, VY11}. In their setting the promise was that $w=Mx$ in the \YES case and $w=Mx\oplus 1^r$ in the \NO case.}
Alice sends a message $m$ to Bob, who must distinguish between the two cases above.  

\paragraph{Distributional Boolean Hidden Partition Problem (D-BHP).}
We will evaluate protocols for this problem on the distribution where  
(1) Alice's input $x$ is uniformly random in $\bool^n$; (2) Bob's graph is sampled from the distribution $\G_{n, \alpha/n}$ and (3) the answer is \YES with probability $1/2$ and \NO with probability $1/2$, independent of Alice's input. We will say that an algorithm achieves advantage $\delta$ over random guessing for the \DBHP problem if it succeeds with probability at least $1/2+\delta$ over the randomness of the input distribution. We will be interested in the
one-way communication complexity of protocols that achieve advantage 
$\delta$ for the \DBHP problem as a function of the parameters $n, \alpha$ and $\delta$.  For technical reasons instead of using the parameters $n, \alpha, \delta$ it will be more convenient to introduce an auxliary parameter $\gamma$. We will prove that any protocol that achieves advantage $\gamma+\alpha^{3/2}>0$ over random guessing for \DBHP requires at least $\Omega(\gamma \sqrt{n})$ communication (the parameter $\alpha$ appears in the expression for the advantage that the protocol is assumed to get due to the possible presence of cycles in Bob's graph $G$).

The rest of the section is devoted to proving 
\begin{lemma}\label{lm:lb-final}
Let $\gamma\in (\lba, 1)$ and $\alpha\in (\lba, 1/16)$ be  parameters. 
Consider an instance of the \DBHP problem where Alice receives a uniformly random string $x\in \bool^n$, and Bob receives a graph $G\in \G_{n, \alpha/n}$, together with the corresponding vector $w$.  No protocol for the \DBHP problem that uses at most $\gamma\sqrt{n}$ communication can get more than $O(\gamma+\alpha^{3/2})$ advantage over random guessing
when inputs are drawn from this distribution.
\end{lemma}
The proof of Lemma~\ref{lm:lb-final} is the main result of this section. The proof follows the outline of ~\cite{GKKRW07, VY11}. 
One crucial difference is that we are working with Erd\H{o}s-R\'{e}nyi graphs as opposed to matchings. This requires replacing some components in the proof. For example, we need to prove a new bound on the expected contribution of Fourier coefficients of a typical message to the distribution of $Mx$ as a function of the weight of these coefficients (Lemmas~\ref{lm:paths} and~\ref{lm:weight}). We also need to take into account the fact that cycles, which are unlikely in sufficiently sparse random graphs, can still arise. This leads to an extra term of $\alpha^{3/2}$ in the statement of Lemma~\ref{lm:lb-final}, and requires a careful choice of the parameter $\alpha$ in the proof of Theorems~\ref{thm:main} and~\ref{thm:main-iid}.
We first give definitions and an outline of the argument, and then proceed to the technical details.

Alice's messages induce a partition $A_1, A_2,\ldots, A_{2^c}$ of $\bool^n$, where $c$ is the bit length of Alice's message. First, a simple argument shows that most strings $x\in \bool^r$ get mapped to `large' sets in the partition induced by Alice's messages. We then show  (see Lemma~\ref{lm:l2-dist}) that if $x$ is uniformly random in such a typical set $A_i\subseteq \bool^n$, the distribution of $Mx$ is close to uniform over $\bool^r$, again for a `typical' graph $G$ received by Bob (and hence a typical edge incidence matrix $M$). We then note that the {\bf BHP} problem can be viewed as Bob receiving a sample from one of two distributions: either $Mx$ (\YES case) or $UNIF(\bool^r)$ (\NO case). Since we showed that the distribution of $Mx$ as $x$is uniform in $A_i$ is close to uniform, implying that it is impossible to distinguish between the two cases from one sample with sufficient certainty. Our main contribution here is the extension of the techniques of ~\cite{GKKRW07, VY11} to handle the case when Bob's input is a subcritical Erd\H{o}s-R\'{e}nyi graph as opposed to a matching. This requires replacing some components in the proof. For example, we need to prove a new bound on the expected contribution of Fourier coefficients of a typical message to the distribution of $Mx$ as a function of the weight of these coefficients (Lemmas~\ref{lm:paths} and~\ref{lm:weight}). We also need to take into account the fact that cycles, which are unlikely in sufficiently sparse random graphs, can still arise. This leads to an extra term of $\alpha^{3/2}$ in the statement of Lemma~\ref{lm:lb-final}, and requires a careful choice of the parameter $\alpha$ in the proof of Theorems~\ref{thm:main} and~\ref{thm:main-iid}.
 We now proceed to give the technical details.

As mentioned above, Alice's messages induce a partition $A_1, A_2,\ldots, A_{2^c}$ of $\bool^n$, where $c$ is the bit length of Alice's message. Since there are $2^c$ such sets, at least a $1-\gamma/2$ fraction of 
$\bool^n$ is contained in sets $A_i$ whose size is at least $(\gamma/2) 2^{n-c}$.  Since our protocol achieves advantage at least $\gamma$ over random guessing on the input distribution, it must achieve advantage at least $1/2+\gamma/3$ conditional on Alice's vector $x$ belonging to one of such large sets. Fix such a set $A\subset \{0, 1\}^{n}$, and let $c'=c+\log(2/\gamma)$, so that $|A|\geq 2^{n-c'}$. Let $f:\{0, 1\}^{n}\to \{0, 1\}$ be the indicator of $A$.

Our analysis relies on the properties of the Fourier transform of the function $f$, similarly to~\cite{GKKRW07}. We use the following normalization of the Fourier transform:
$$
\hat f(v)=\frac1{2^n} \sum_{x\in \{0, 1\}^n} f(x) (-1)^{x\cdot v}.
$$

We will use the following bounds on the Fourier mass of $f$ contributed by coefficients of various weight:
\begin{lemma}[Lemma~6 in \cite{GKKRW07}]\label{lm:weight-bounds}
Let $A\subseteq \bool^n$ of size at least $2^{n-c'}$. Then for every $\ell\in \{1, 2,\ldots, 4c'\}$
$$
\frac{2^{2n}}{|A|^2}\sum_{v: |v|=\ell} \hat f(v)^2\leq \left(\frac{4\sqrt{2}c'}{\ell}\right)^\ell.
$$
\end{lemma}

As before, we denote the graph that Bob receives as input by $G$, and the number of edges in $G$ by $r$. Let the edge incidence matrix of $G$ be denoted by $M$, i.e. $M_{eu}=1$ iff $u\in [n]$ is an endpoint of $e\in {[n] \choose 2}$. We are interested in the distribution of $Mx$, where $x$ is uniformly random in $A$. 
For $z\in \{0, 1\}^{r}$ let 
$$
p_M(z)=\frac{|\{x\in A: Mx=z\}|}{|A|}.
$$
Note that $p_M(z)$ is a function of the message $A$. We will supress this dependence in what follows to simplify notation.  This will not cause any ambiguity since $A$ is fixed as a typical large set arising from Alice's partition.  We would like to prove that $p_M(z)$ is close to uniform. We will do that by bounding the Fourier mass in positive weight coefficients of $p_M(z)$. By the same calculation as in ~\cite{GKKRW07} (Lemma~10), we have 
\begin{equation*}
\begin{split}
\wh{p_M}(s)&=\frac1{2^r} \sum_{z\in \{0, 1\}^r} p_M(z) (-1)^{z\cdot s}\\
&=\frac1{|A|2^r} \left(|\{x\in A: (Mx)\cdot s=0\}|-|\{x\in A: (Mx)\cdot s=1\}|\right)\\
&=\frac1{|A|2^r} \left(|\{x\in A: x\cdot (M^Ts)=0\}|-|\{x\in A: x\cdot(M^Ts)=1\}|\right)\\
&=\frac1{|A|2^r} \sum_{x\in \{0, 1\}^n} f(x)\cdot (-1)^{x\cdot (M^Ts)}\\
&=\frac{2^n}{|A|2^r} \wh{f}(M^Ts),\\
\end{split}
\end{equation*}
and
\begin{equation}\label{eq:tv2energy}
\begin{split}
||p_M-U_r||^2_{tvd}&\leq 2^r||p_M-U_r||^2_{2}\\
&= 2^{2r}\sum_{s\in \bool^r, s\neq 0} 	\wh{p}_M(s)^2\\
&=\frac{2^{2n}}{|A|^2}\sum_{s\in \bool^r, s\neq 0} 	\wh{f}(M^Ts)^2.
\end{split}
\end{equation}
Here the first transition is by Cauchy-Schwartz, the second is Parseval's equality, and $U_r$ is the uniform distribution over $\bool^r$, which we also denote by $UNIF(\bool^r)$.

It is convenient to write 
\begin{equation}\label{eq:weight-classes}
\begin{split}
||p_M-U_r||^2_{tvd}&\leq \sum_{s\in \{0, 1\}^r, s\neq 0^r} \wh{f}(M^Ts)^2\\
&=\sum_{v\in \{0, 1\}^n} \wh{f}(v)^2\cdot |\{s\in \{0, 1\}^r, s\neq 0, v=M^Ts\}|\\
&=\sum_{\ell\geq 0} \sum_{v\in \{0, 1\}^n, wt(v)=\ell} \wh{f}(v)^2\cdot |\{s\in \{0, 1\}^r, s\neq 0, v=M^Ts\}|
\end{split}
\end{equation}

We note that the vector $s\in \bool^r$  assigns numbers in $\bool$ to edges of Bob's graph $G$ (the interpretation of $s$ as a vector in $\bool^r$ requires an implicit numbering of edges; this numbering is implicitly defined by the incidence matrix $M\in \bool^{r\times n}$). The analysis to follow will bound the expectation of the summation on the lhs of \eqref{eq:weight-classes} with respect to the edge incidence matrix $M$ of an Erd\H{o}s-R\'{e}nyi graph. In order to achieve this, we will bound the expecation of the rhs of \eqref{eq:weight-classes}.  crucially using the interplay between two bounds. First,  we will prove that the expected (over the random graph $G$, and hence its edge incidence matrix $M$) number of representations of a vector $v\in \bool^n$ as $M^T s$ for $s\in \bool^r$ decays  with the weight of $v$. On the other hand, the amount of $\ell_2$ mass on the Fourier coefficients $\wh{f}(v)$ of given weight $\ell$ does not grow too fast as a function of the weight class by Lemma~\ref{lm:weight-bounds}. Before proceeding to the proof, we summarize relevant notation.

\paragraph{Notation} The set of Alice's inputs that correspond to a typical message is denoted by $A\subseteq \bool^n$, its indicator function is denoted by $f:\bool^n\to \bool$.  Bob's graph, which is sampled from the distribution $\G_{n, \alpha/n}$, is denoted by $G$, its edge incidence matrix is denoted by $M\in \bool^{r\times n}$. By \eqref{eq:weight-classes}, in order to bound the distance from $p_M$ to uniformity, it is sufficient to bound the $\ell_2$ norm of the nonzero weight part of the Fourier spectrum of $f$.  For each Fourier coefficient $\wh{f}(v)$ we need to bound the number of ways of representing  $v\in \bool^n$ as $M^Ts, s\in \bool^r$. We will bound this quantity in terms of the weight of $v$. In order to prove such a bound, we start by showing a structural property of vectors $s\in \bool^r$ that satisfy $v=M^Ts$ for a given $v\in \bool^n$:

\begin{lemma}\label{lm:paths}
Fix $v\in \{0, 1\}^{n}$. Let $s\in \{0, 1\}^r$, and let $F=(V, E_F)$ contain those edges of $G$ that belong to the support of $s$. Then $s\in \{0, 1\}^r, s\neq 0^r$ satisfies $v=M^Ts$ if and only if $F$ is an edge-disjoint union of paths connecting pairs of nonzero elements of $v$ and cycles.  In particular, $v$ must have even weight. If $G$ contains no cycles, the weight of $v$ must be positive.
\end{lemma}
\begin{proof}
Note that $M^Ts$ is the sum of incidence vectors of edges whose values in $s$ are nonzero. Since $M^Ts=v$, it must be that all vertices $i\in V$ such that $v_i=0$ have even degrees in the subgraph $F$, and all vertices $i$ with $v_i=1$ have odd degrees. Thus implies that the edge set if $F$ can be decomposed into a union of edge-disjoint paths that connect nonzeros in $v$ and a disjoint union of cycles, as required. In particular, $v$ must have even weight, strictly positive if $G$ contains no cycles.
\end{proof}

We now fix $v\in \{0, 1\}^n$ of even weight $\ell$ and bound the quantity $\expect_M\left[|\{s\in \{0, 1\}^r, v=M^Ts\}|\right]$. More precisely, in Lemma~\ref{lm:weight} below we only bound a related quantity, in which we exclude $s$ that contains cycles from consideration. The case of cycles is handled directly in the proof of our main lemma that bounds the distance of $p_M$ to uniformity (Lemma~\ref{lm:l2-dist}).
\begin{lemma}\label{lm:weight}
Let $v\in \{0, 1\}^n$ have even weight $\ell$. Let $G$ be a random graph sampled according to $\G_{n, \alpha/n}$, and let $M\in \bool^{r\times n}$ be its edge incidence matrix. 
Then 
$$
\expect_M\left[|\{s\in \{0, 1\}^r, s\neq 0^r, v=M^Ts, \text{~s~is a union of edge-disjoint paths}\}|\right]\leq 2^\ell (\ell/2)! (C\alpha /n)^{\ell/2}
$$
\end{lemma}
\begin{proof}

By Lemma~\ref{lm:paths} if $v=M^Ts$, it must be that $s$ is a union of edge-disjoint paths and cycles, where the endpoints of the paths are nonzeros of $v$.  Thus, we are interested in unions of paths $P_i, i=1,\ldots, \ell/2$, connecting nonzeros of $v$.  

We now fix a pairing of nonzeros of $v$. For notational simplicty, assume that path $P_i$ connects the $(2i-1)$-st nonzero of $v$ to the $2i$-th for $i=1,\ldots, \ell/2$. 
For one such path $P_i$  one has $\prob[P_i\subseteq G]=(\alpha/n)^q$, where $q$ is the length of $P_i$.  By a union bound over all path lengths $q\geq 1$ and all paths of length $q$ connecting the $(2i-1)$-st nonzero to the $2i$-th nonzero we have
\begin{equation*}
\begin{split}
&\prob[P_i\subseteq G]\leq \sum_{q\geq 1}  n^{q-1}\cdot (\alpha/n)^{q}\leq C\alpha /n\\
\end{split}
\end{equation*}
for a constant $C>0$.  Since the paths are edge disjoint, we have 
\begin{equation}\label{eq:sub-mult}
\begin{split}
&\prob[P_i\subseteq G\text{~for all~}i=1,\ldots, \ell/2]\leq \prod_{i=1}^{\ell/2} \prob[P_i\subseteq G]\leq (C\alpha /n)^{\ell/2}.\\
\end{split}
\end{equation}

It remains to note that there are no more than $2^\ell (\ell/2)!$ ways of pairing up the nonzeros of $v$. Putting this together with ~\eqref{eq:sub-mult}  yields 
\begin{equation*}
\begin{split}
\expect_M\left[|\{s\in \{0, 1\}^r, s\neq 0^r, \right.&\left.v=M^Ts,\text{~$s$~is a union of edge-disjoint paths}\}|\right]\\
&\leq 2^\ell (\ell/2)!\cdot\sum_{P_1,P_2,\ldots, P_{\ell/2}}  \prob[P_i\subseteq G\text{~for all~}i=1,\ldots, \ell/2]\leq 2^\ell (\ell/2)!\cdot (C\alpha/n)^{\ell/2}\\
\end{split}
\end{equation*}
as required.

\end{proof}

Equipped with Lemma~\ref{lm:weight}, we can now prove that $p_M$ is close to uniform:
\begin{lemma}\label{lm:l2-dist}
Let $A\subseteq \bool^n$ of size at least $2^{n-c'}$, and let $f:\bool^n\to \{0, 1\}$ be the indicator of $A$. Let $G$ be a random graph sampled according to $\G_{n, \alpha/n}$, where $\alpha\in (\lba, 1/16)$ is smaller than an absolute constant.  Suppose that $c'\leq \gamma \sqrt{n}+\log(2/\gamma)$ for some $\gamma\in (\lba, 1)$ smaller than an absolute constant. Then 
\begin{equation*}
\expect_M[||p_M-U_r||^2_{tvd}]=O(\gamma^2+\alpha^3).
\end{equation*}
\end{lemma}
\begin{proof}
Let $\E$ denote the event that the graph $G$ contains no cycles and the number of edges $r$ in $G$ is at most $2\alpha n$. By Lemma~\ref{lm:cycles}, Lemma~\ref{lm:complex-comp} and Chernoff bounds we have $\prob[\E]\geq 1-O(\alpha^3)-e^{-\Omega(\alpha n)}\geq 1-O(\alpha^3)$ as long as $\alpha\geq \lba$.
We have
\begin{equation*}
\begin{split}
\expect_M[||p_M-U_r||^2_{tvd}]\leq \expect_M[||p_M-U_r||^2_{tvd}|\E]\prob[\E]+\prob[\bar \E]\leq \expect_M[||p_M-U_r||^2_{tvd}|\E]\prob[\E]+O(\alpha^3)
\end{split}
\end{equation*}
since $||p_M-U_r||_{tvd}\leq 1$ always. We now concentrate on bounding the first term, i.e. we are bounding the expectation of $||p_M-U_r||_{tvd}^2$ conditional on $G$ having no cycles. 

By~\eqref{eq:weight-classes} we have
\begin{equation}\label{eq:fouvb}
\begin{split}
\expect_M[||p_M-U_r||^2_{tvd}|\E]&\leq \frac{2^{2n}}{|A|^2}\expect_M\left[\sum_{s\in \bool^r} 	\wh{f}(M^Ts)^2|\E\right]\\
&= \frac{2^{2n}}{|A|^2}\sum_{x\in \{0, 1\}^n} \wh{f}(v)^2\cdot \expect_M\left[|\{s\in \{0, 1\}^r, v=M^Ts\}||\E\right]\\
&= \frac{2^{2n}}{|A|^2}\sum_{\text{even~}\ell\geq 0}\sum_{v\in \{0, 1\}^n\text{~of weight~}\ell} \wh{f}(v)^2\cdot \expect_M\left[|\{s\in \{0, 1\}^r, v=M^Ts\}||\E\right]\\
&= \frac{2^{2n}}{|A|^2}\sum_{\text{even~}\ell\geq 0}^{4\alpha n}\sum_{v\in \{0, 1\}^n\text{~of weight~}\ell} \wh{f}(v)^2\cdot \expect_M\left[|\{s\in \{0, 1\}^r, v=M^Ts\}||\E\right]\\
\end{split}
\end{equation}
Here the first three lines are by taking conditional expectations in \eqref{eq:weight-classes}, and the restriction in the summation in last line follows by conditioning in $\E$: since the number of edges in $G$ conditional on $\E$ is at most $2\alpha n$,  the weight of $M^Ts$ is bounded by $4\alpha n$ for all $s\in \bool^r$. 

We further break this summation into two parts $\ell\in [0:4c')$ and $\ell\in [4c':4\alpha n]$. 
We first consider $\ell\in [0:4c)$. First, by Lemma~\ref{lm:paths} and conditioning on $\E$ (i.e. that $G$ does not contain cycles), it is sufficient to consider $\ell>0$. Further, we have
\begin{equation*}
\begin{split}
&\frac{2^{2n}}{|A|^2}\sum_{\text{even~}\ell=2}^{4c'-2}\sum_{v\in \{0, 1\}^n\text{~of weight~}\ell} \wh{f}(v)^2\cdot \expect_M\left[|\{s\in \{0, 1\}^r, v=M^Ts\}||\E\right]\\
&=\frac{2^{2n}}{|A|^2}\sum_{\text{even~}\ell=2}^{4c'-2}\sum_{v\in \{0, 1\}^n\text{~of weight~}\ell} \wh{f}(v)^2\cdot \expect_M\left[|\{s\in \{0, 1\}^r, v=M^Ts, s\text{~a union of edge-disjoint paths}\}||\E\right]\\
&\leq\frac1{\prob[\E]}\frac{2^{2n}}{|A|^2}\sum_{\text{even~}\ell=2}^{4c'-2}\sum_{v\in \{0, 1\}^n\text{~of weight~}\ell} \wh{f}(v)^2\cdot \expect_M\left[|\{s\in \{0, 1\}^r, v=M^Ts, s\text{~a union of edge-disjoint paths}\}|\right],\\
\end{split}
\end{equation*}
where the first transition follows by Lemma~\ref{lm:paths} using the fact that  $G$ contains no cycles conditional on $\E$. We are now in the setting of Lemma~\ref{lm:weight}. Using Lemma~\ref{lm:weight} and Lemma~\ref{lm:weight-bounds}, we get
\begin{equation}\label{eq:dfioshv}
\begin{split}
\frac1{\prob[\E]}\frac{2^{2n}}{|A|^2}\sum_{\text{even~}\ell=2}^{4c'-2}\sum_{v\in \{0, 1\}^n\text{~of weight~}\ell} \wh{f}(v)^2&\cdot \expect_M\left[|\{s\in \{0, 1\}^r, v=M^Ts, s\text{~a union of edge-disjoint paths}\}|\right]\\
&\leq\frac1{\prob[\E]}\sum_{\text{even~}\ell=2}^{4c'-2} 2^\ell (\ell/2)!(C\alpha /n)^{\ell/2} \frac{2^{2n}}{|A|^2}\sum_{v\in \{0, 1\}^n\text{~of weight~}\ell} \wh{f}(v)^2\\
&\leq \frac1{\prob[\E]}\sum_{\text{even~}\ell=2}^{4c'-2}\frac{(\ell/2)! (C \alpha)^{\ell/2}}{n^{\ell/2}}\left(\frac{4\sqrt{2} c'}{\ell}\right)^{\ell}\\
&\leq \frac1{\prob[\E]}\sum_{\text{even~}\ell=2}^{4c'-2}\frac{(4C \alpha)^{\ell/2}}{(n/\ell)^{\ell/2}}\left(\frac{4\sqrt{2} c'}{\ell}\right)^{\ell}\\
&\leq \frac1{\prob[\E]}\sum_{\text{even~}\ell=2}^{4c'-2}((c')^2/n)^{\ell/2}\left(C'/\ell\right)^{\ell/2}\\
&=O(\gamma^2)
\end{split}
\end{equation}
whenever $c'\leq \gamma \sqrt{n}+\log(2/\gamma)$ and $\gamma\geq \lba$ (which is satisfied by the assumptions of the lemma). Here the second line is by Lemma~\ref{lm:weight}, the third line is by Lemma~\ref{lm:weight-bounds} and the fourth and fifth lines are by algebraic manipulation.  Note that we used the fact that $\frac1{\prob[\E]}=1/(1-O(\alpha^3))=O(1)$.

We now consider the range $\ell\in [4c':4\alpha n]$.  Note that the function $2^\ell (\ell/2)!\left(\frac{C\alpha}{n}\right)^{\ell/2}$ is decreasing in $\ell$ for $\ell<n/4$
$$
\frac{2^\ell (\ell/2)!(C\alpha)^{\ell/2}/n^{\ell/2}}{2^{\ell+2} ((\ell+2)/2)!(C\alpha)^{(\ell+2)/2}/n^{(\ell+2)/2}}=\frac{n}{4(\ell+1)}\geq 1
$$
since $C\alpha<1$.
Since $\ell\leq 4\alpha n<n/4$, we have
\begin{equation}\label{eq:234sdc}
\begin{split}
&\frac{2^{2n}}{|A|^2}\sum_{\text{even~}\ell=4c'}^{4\alpha n}\sum_{v\in \{0, 1\}^n\text{~of weight~}\ell} \wh{f}(v)^2\cdot \expect_M\left[|\{s\in \{0, 1\}^r, v=M^Ts\}||\E\right]\\
&\leq \frac1{\prob[\E]}2^{4c'} ((4c')/2)!(C\alpha /n)^{4c'/2}\leq (4c')^{4c'}(C\alpha /n)^{4c'/2}=\frac1{\prob[\E]}(C\alpha 16(c')^2 /n)^{4c'/2}=O(\gamma^2)\\
\end{split}
\end{equation}
since $c'\geq 1$.
Putting \eqref{eq:fouvb} together with \eqref{eq:dfioshv} and \eqref{eq:234sdc} completes the proof.
\end{proof}
We will also need the following simple lemma, whose proof is given in Appendix~\ref{app:reduction}:
\begin{lemma}\label{lm:tvd-cond}
Let $(X, Y^1), (X, Y^2)$ be random variables taking values on finite sample space $\Omega=\Omega_1\times \Omega_2$. For any $x\in \Omega_1$ let $Y^i_x, i=1,2$ denote the conditional distribution of $Y^i$ given the event $\{X=x\}$. Then 
$$
||(X, Y^1)-(X, Y^2)||_{tvd}=\expect_X[||Y^1_X-Y^2_X||_{tvd}].
$$
\end{lemma}
We can now prove the main result of this section, namely that no algorithm that uses $o(\sqrt{n})$ can get substantial advantage over random guessing for \DBHP:

\newcommand{\A}{{{P}}}
\newcommand{\B}{{{Q}}}

\noindent\begin{proofof}{Lemma~\ref{lm:lb-final}}

Let $\A(x)$ denote the function Alice applies to her input $x$ to compute the 
message to send to Bob. Let $\B(M,i,w)$ be the Boolean function computed by 
Bob on his inputs $M,w$ and message $i$ from Alice. (Without loss of 
generality, since we are proving hardness against a fixed distribution, 
$\A$ and $\B$ are deterministic; also, we write use the edge incidence matrix $M$ instead of graph $G$ to denote Bob's input). Let $D^1$ be the distribution of $(M,\A(x),w)$ on 
\YES instances and $D^2$ be the distribution of $(M,\A(x),w)$ on \NO instances. 
We now show that $||D^1 - D^2||_{tvd}=O(\gamma+\alpha^{3/2})$, which
by definition of total variation distance implies that the protocol has advantage at most 
$O(\gamma+\alpha^{3/2})$ on \DBHP.

Alice's function $\A(x)$ induces a partition $A_1, A_2,\ldots, A_{2^c}$ of $\bool^n$, where $c$ is the bit length of Alice's message. Since there are $2^c$ such sets,  at least a $1-\gamma/2$ fraction of $\bool^n$ is contained in sets $A_i$ whose size is at least $(\gamma/2) 2^{n-c}$. We call a message $i$ such that $|A_i|<(\gamma/2) 2^{n-c}$ {\em typical}.  Say that $x$ is bad if $\A(x)$ is not typical and say that $i = \A(x)$ is 
typical if $x$ is typical. We have that $i=\A(x)$ is not typical with probability at most $\gamma/2$.  Note that
distribution $D^1= (M,i, p_{M, i})$, where $p_{M, i}$ is the distribution of $Mx$ conditional on the message $i$, and  $D^2 = (M,i, U_r)$. 
For any $M, i$ let $D^1_{(M, i)}=p_{M, i}$ and $D^2_{(M, i)}=U_r$ denote the distribution of $w$ given  message $i$ and the matrix $M$.

We have, using Lemma~\ref{lm:tvd-cond} twice,
\begin{equation}\label{eq:3oirnfwer}
\begin{split}
||D^1 - D^2||_{tvd}&=\expect_{i}\left[\expect_{M}\left[||D^1_{(M,i)} - D^2_{(M,i)}||_{tvd}\right]\right]\\
&\leq \prob[i\text{~is typical}]+\expect_{M}\left[||p_{M, i'} - U_r||_{tvd}\right] \text{~~($i'$ any typical message)}\\
&\leq \gamma/2+\expect_{M}\left[||p_{M, i'} - U_r||_{tvd}\right] \text{~~($i'$ any typical message)}.\\
\end{split}
\end{equation}

Suppose that the  protocol uses $c\leq \gamma \sqrt{n}$ bits of communication. Let $i'$ be a typical message. Then we have by definition of a typical set above $|A_{i'}|\geq 2^{n-c'}$ for $c'\leq \gamma \sqrt{n}+\log(2/\gamma)$. Thus, by Lemma~\ref{lm:l2-dist} applied to $A_{i'}$ we have
\begin{equation}
\begin{split}
\expect_{M}\left[||p_{M, i'} - U_r||^2_{tvd}\right]=O(\gamma^2+\alpha^{3}),
\end{split}
\end{equation}
and hence by Jensen's inequality 
\begin{equation}
\begin{split}
\expect_{M}\left[||p_{M, i'} - U_r||_{tvd}\right]\leq \sqrt{\expect_{M}\left[||p_{M, i'} - U_r||^2_{tvd}\right]}=O(\gamma+\alpha^{3/2}).
\end{split}
\end{equation}
Putting this together with \eqref{eq:3oirnfwer}, we get
\begin{equation}
\begin{split}
||D^1 - D^2||_{tvd}&=O(\gamma+\alpha^{3/2})
\end{split}
\end{equation}
as required.
\end{proofof}

%!TEX root = ./factor2.tex
\section{Reduction from \DBHP to MAX-CUT}
\label{sec:reduction}

%\xxx{Give overview. I think that overview is there -- MK. But perhaps it could be improved.}

This section is devoted to proving the following reduction from \DBHP problem to max-cut problem.
\begin{lemma}\label{lm:reduction}
Suppose there is a (single-pass) streaming algorithm \ALG that can distinguish between the \YES and the \NO instances from the distribution $\D$
using space $c$ when edges appear in the canonical random order, with failure probability bounded by $1/10$. Then there exists a protocol for the \DBHP  problem that uses at most $c$ bits of communication and succeeds with probability at least $1/2+\Omega(1/k)$, where $k$ is the number of phases in the distribution $\D$. 
\end{lemma}

%This section is devoted to presenting a reduction from the \DBHP problem to the problem of obtaining a $(2-\e)$-approximation to max-cut value. The {\bf Boolean Hidden Matching} problem defined and studied in~\cite{GKKRW07} is a special case of {\bf Boolean Hidden Partition}.
We start by outlining the connection between \DBHP and our hard distribution for MAX-CUT, and then proceed to give the formal reduction. The connection between \DBHP and the hard distribution $\D$ for approximating maxcut that we defined in section~\ref{sec:dist} is as follows. Suppose that  Alice gets a random string $x\in \bool^n$, and Bob gets a graph $G$ sampled from the distribution $\G_{n, \alpha/n}$, as well as the corresponding vector $w$. Then Bob can use his input to generate a graph $G'$ by including edges $e$ of $G$ that satisfy $w_e=1$. In the \YES case of the communication problem we have $w=Mx$, so if we interpret $x$ as encoding a bipartition of $V$, an edge $e$ satisfies $w_e=1$ iff it crosses the bipartition $R$. Thus, $G'$ has exactly the distribution of a graph $G'_i$ defined in $\D^Y$, i.e. the \YES case distribution on maxcut instances. Similarly, in the \NO case the graph $G'$ generated by Bob using $G$ and $w$ has exactly the distribution of $G_i'$ defined in $\D^N$. The only difference between the communication complexity setting and the distribution $\D$ is that $\D$ naturally consists of several phases, while the communication problem asks for a single-round one-way communication protocol. Nevertheless, in this section we show how any algorithm for max cut that succeeds on $\D$ can be converted into a protocol for \DBHP.

In what follows we assume existence of an algorithm \ALG that yields a $(2-\e)$-approximation to maxcut value in a graph on $n$ nodes using space $c$ and a single pass over a stream of the edges of the graph given in a random order, and failure probability bounded by $1/10$.  We show how \ALG can be used to obtain a protocol for \DBHP  that uses at most $c$ bits of communication and gives advantage at least $\Omega(\e^2/\log n)$. In what follows we will only evaluate the performance of \ALG on the distribution $\D$. By Yao's minimax principle~\cite{Yao87}, there exists a {\em deterministic} algorithm {\bf ALG}' that errs with probability at most $1/10$ on inputs drawn from $\D$. We will work with deterministic algorithms in what follows. To simplify notation, we will simply assume that \ALG is a deterministic algorithm from now on that distinguishes between the \YES and \NO instances generated by $\D$ with probability at least $9/10$.

In order to describe the reduction from \DBHP to MAX-CUT, we first recall how the distribution $\D$ over MAX-CUT instances was defined. To sample a graph from $\D$, one first samples a uniformly random partition $R=(P, Q)$ of the vertex set $V=[n]$ (see section~\ref{sec:input-dist}). Then graphs $G_1=(V, E_1), G_2=(V, E_2),\ldots, G_k=(V, E_k)$ are sampled from the distribution $\G_{n, \alpha/n}$.  These graphs are auxiliary, and only their carefully defined subgraphs $G_i'$ appear in the stream. The subgraphs $G_i'$ are defined differently in the \YES and \NO cases. In the \YES case $G'_j=(V, E'_j)$ contains all edges of $G$ that cross the bipartition, i.e. those edges that satisfy $(Mx)_e=1$. In the \NO case, $G'_j$ contains each edge of $G_j$ independently with probability $1/2$. 

We now define random variables that describe the execution of \ALG on $\D$. These random variables will be crucially used in our reduction. Let the state of the memory of the algorithm after receiving the subset of edges $E'_i, i=1,\ldots, k$ be denoted by $S^{Y}_i$ and $S^N_i$ in the \YES and \NO case respectively. Thus, $S^Y_i, S^N_i\in \bool^c$ for all $i$. We assume wlog that $S^Y_0=S^N_0=0$, since the algorithm \ALG starts in some  fixed initial configuration. 

Note that the main challenge that we need to overcome in order to reduce \DBHP to MAX-CUT on our distribution $\D$ is that while \DBHP is a two-party one-way communication problem, the distribution $\D$ inherently consists of $k$ `phases'. Intuitively, we overcome this difficulty by showing that a successful algorithm for solving maxcut on $\D$ must solve \DBHP {\em in at least one of the $k$ phases}. Our main tool in formalizing this intuition is the notion of an {\em informative index}:
\begin{definition}[Informative index]
We say that an index $j\in \{1, \ldots, k\}$ in the execution of the algorithm \ALG is $\Delta$-{\em informative} for $\Delta>0$ if 
$||S^Y_{j+1}-S^{N}_{j+1}||_{tvd}\geq ||S^Y_{j}-S^{N}_{j}||_{tvd}+\Delta$.
\end{definition}

The next lemma shows that for any \ALG that distinguishes between the \YES and \NO cases with probability at least $9/10$ over inputs from $\D$, there exists an $\Omega(1/k)$-informative index. The proof is based on a standard hybrid argument and is given below.

\begin{lemma}\label{lm:inf}
For any algorithm \ALG that succeeds with probability at least $9/10$ on inputs drawn from $\D$, there exists a $\Omega(1/k)$-informative index.
\end{lemma}
\begin{proof}
The proof is essentially the standard hybrid argument.
First note that since the algorithm starts in some fixed state, we have $||S^Y_0-S^N_0||_{tvd}=0.$
On the other hand, since \ALG distinguishes between the \YES and \NO cases with probability at least $9/10$ on inputs drawn from $\D$, we must have $||S^Y_k-S^N_k||_{tvd}\geq C$
for a constant $C>0$.   Let $j$ be the smallest integer such that $||S^Y_{j+1}-S^{N}_{j+1}||_{tvd}\geq C (j+1)/k$.
By this choice of $j$ we have $||S^Y_{j}-S^{N}_{j}||_{tvd}<C j/k$.
 Thus, $||S^Y_{j+1}-S^{N}_{j+1}||_{tvd}-||S^Y_{j}-S^{N}_{j}||_{tvd}\geq C (j+1)/k-C j/k\geq C/k$ as required. 
\end{proof}

We now fix a $\Delta$-informative index $j^*\in [1:k-1]$. We will show below that Alice and Bob can get $\Delta$ advantage over random guessing for the \DBHP problem using \ALG, thus completing the reduction from \DBHP to MAX-CUT.

Recall that in \DBHP, Alice gets a uniformly random $x\in \bool^n$ as input, Bob gets a graph $G=(V, E), V=[n], E\subseteq {V \choose 2}, |E|=r$ sampled from $\G_{n, \alpha/n}$ and a vector $w\in \bool^r$. in the \YES case of \DBHP the vector $w$ satisfies $w=Mx$ and in the \NO case $w$ is uniformly random in $\bool^r$. Here $M\in \bool^{r\times n}$ is the edge incidence matrix of $G$, i.e. $M_{ev}=1$ if $v\in V=[n]$ is an endpoint of $e\in E$ and $M_{ev}=0$ otherwise. The \YES case occurs independently with probability $1/2$, and the \NO case occurs with remaining prabability.

In order to reduce to MAX-CUT, we view $x$ as encoding a partition $R=(P, Q)$ of the vertex set $V=[n]$. With this interpretation, 
the vector $w$ that Bob gets assigns numbers $w_e\in \{0, 1\}$ to edges of $G$. In the \YES case these numbers encode whether or not the edge $e$ crosses the bipartition $R$ encoded by $x$, and in the \NO case these numbers are uniformly random and independent. 
This connection lets Alice and Bob draw an input instance for MAX-CUT from distribution $\D^Y$ or $\D^N$, depending on the answer to their \DBHP problem from their inputs $x$ and $(G, w)$ as follows:

\begin{description}
\item[Step 1.] Alice samples a state $s$ of \ALG from the distribution $S^Y_{j^*}$. She can do that since she knows $x$. Indeed, first Alice   generates random graphs $G_{i}=(V, E_i), i=1,\ldots, j^*$ from the distribution $\G_{n, \alpha/n}$ as specified in the definition of $\D^Y$ and computes $w_i=M_i x\in \bool^{r_i}$ for the edge incidence matrix $M_i\in \bool^{r_i\times n}$ of $G_i$ ($r_i$ denotes the number of edges in $G_i$). She then lets $G_i'=(V, E_i')$ contain the set of edges $e\in E_i$ of $G_i$ that satisfy $(w_i)_e=1$ (i.e. cross the bipartition encoded by $x$). Alice then runs \ALG on the stream of edges $E'_1,\ldots, E'_{j^*}$, where edges inside each $E'_i$ are presented to the algorithm in a uniformly random order.  She then sends the state $s$ of \ALG to Bob. 

\item[Step 2.] Bob first creates a graph $G'$ by including those edges $e$ of his input graph $G$ that satisfy $w_e=1$ (recall that Bob's input is the pair $(G, w)$). He then creates the random variable $\tilde s$ by running \ALG for one more step starting from $s$ on $G'$. We let $G'_{j^*+1}:=G'$ for convenience. Denote the distribution of $\tilde s$ in the \YES case by $\tilde S^Y$ and the distribution of $\tilde s$ in the \NO case by $\tilde S^N$.
 \item[Step 3.] Bob outputs \YES if $p_{\tilde S^{Y}}(\tilde s)>p_{\tilde S^N}(\tilde s)$ and \NO otherwise. 
 \end{description}

Note that the distribution $\tilde S^Y$ is identical to $S^Y_{j^*+1}$. The distribution $\tilde S^N$ is in general different from both $S^Y_{j^*+1}$ and $S^N_{j^*+1}$, however.  
We first show that the protocol above is feasible:
\begin{claim} 
Steps 1-3 above give a valid protocol for the \DBHP communication problem.
\end{claim}
\begin{proof}
Steps 1 and 2 are feasible, as shown above. The distribution $\tilde S^N$ that we constructed can be generated as follows: pick a random $x\in \{0, 1\}^n$, run the algorithm on that $x$ assuming it is the \YES case for $j^*$ steps,  then take one \NO step (the $j^*+1$-st).  Bob can compute the pdf of this distribution. Similarly for $\tilde S^Y$.  Thus, Alice and Bob can execute the protocol.
 \end{proof}

We now prove Lemma~\ref{lm:reduction}, which is the main result of this section.  For that, we will need the following auxiliary claim, whose proof is given in Appendix~\ref{app:reduction}.

\begin{claim}\label{cl:1}
Let $X, Y$ be two random variables. Let $W$ be independent of $(X, Y)$. Then for any function $f$ one has 
$||f(X, W)-f(Y, W)||_{tvd}\leq ||X-Y||_{tvd}$.
\end{claim}

We will also need the following lemma, which says that if one receives a sample from one of two distributions $X$ and $Y$ on a finite probability space $\Omega$, a simple test distinguishes between the two distributions with advantage at least $||X-Y||_{tvd}/2$ over random guessing.

\begin{lemma}\label{lm:single-sample}
Let $X, Y$ be distributions on a finite probability space $\Omega$.  Suppose that with probability $1/2$ one is given a sample $\omega$ from $X$ (\YES case) and with probability $1/2$ a sample from $Y$ (\NO case). Then outputting \YES if $p_X(\omega)>p_Y(\omega)$ and \NO otherwise distinguishes between the two cases with advantage over random guessing at least $||X-Y||_{tvd}/2$.
\end{lemma}
The proof is given in Appendix~\ref{app:reduction}.

\begin{proofof}{Lemma~\ref{lm:reduction}}
As before, we assume that \ALG is deteministic. Let $j^*$ be an informative index for \ALG, which exists by Lemma~\ref{lm:inf}.

Let $f:\{0, 1\}^c\times \bool^{{n \choose 2}} \to \{0, 1\}^c$ denote the function that maps the state of \ALG at step $j^*$ and the edges received at step $j^*+1$ to the state of \ALG at step $j^*+1$. Let $G'_{j^*+1}$ denote the set of edges $e$ of $G_{j^*+1}$ that satisfy $(w_{j^*+1})_e=1$. By {\bf Step 2} of our reduction we have 
$\tilde s=f(s, G'_{j^*+1})$, and hence $\tilde S=f(S^Y_{j^*}, G'_{j^*+1})$.

Suppose that we are in the \NO case. Then Bob's input $G'_{j^*+1}$ is a random graph sampled independently from $\G_{n, \alpha/(2n)}$.    Note that $S^{N}_{j^*+1}$ is distributed as $f(S^N_{j^*}, G'_{j^*+1})$, so by  Claim~\ref{cl:1}
\begin{equation}\label{eq:sn}
\begin{split}
||\tilde S^N-S^{N}_{j^*+1}||_{tvd}&=||f(S^Y_{j^*}, G'_{j^*+1})-f(S^{N}_{j^*}, G'_{j^*+1})||_{tvd}\leq ||S^Y_{j^*}-S^{N}_{j^*}||_{tvd}.\\
\end{split}
\end{equation}

Now suppose that we are in the \YES case. Denote the distribution of $\tilde s$ in this case by $\tilde S^Y$. Then $\tilde S^Y=f(S_{j^*}^Y, G'_{j^*+1})=S_{j^*+1}^Y$. Thus,
\begin{equation}\label{eq:sy}
\begin{split}
||\tilde S^Y-S^{N}_{j^*+1}||_{tvd}&=||S^Y_{j^*+1}-S^{N}_{j^*+1}||_{tvd}.
\end{split}
\end{equation}

Putting \eqref{eq:sn} and \eqref{eq:sy} together and using triangle inequality and the assumption that $j^*$ is  $\Omega(\alpha \e^2)$-informative we get
\begin{equation*}
\begin{split}
||\tilde S^Y-\tilde S^N||_{tvd}&\geq ||\tilde S^Y-S^N_{j^*+1}||_{tvd}-||S^N_{j^*+1}-\tilde S^N||_{tvd}\geq ||S^Y_{j^*+1}-S^N_{j^*+1}||_{tvd}-||S^N_{j^*}-S^N_{j^*}||_{tvd} \geq \Omega(1/k).
\end{split}
\end{equation*}

Thus, we are getting one sample from one of two distributions whose total variation distance is at least $\Omega(\alpha \e^2)$. With probability $1/2$ we are getting a sample $\tilde s$ from $\tilde S^Y$ and with probability $1/2$ a sample $\tilde s$ from $\tilde S^N$. By Lemma~\ref{lm:single-sample} the simple algorithm that Bob uses, namely outputting \YES if $p_{\tilde S^Y}(\tilde s)>p_{\tilde S^N}(\tilde s)$ and \NO otherwise yields advantage at least 
$||\tilde S^Y-\tilde S^N||_{tvd}/2\geq \Omega(1/k)$
over random guessing, as required.
\end{proofof}

%!TEX root = ./factor2.tex
\section{$\tilde \Omega(\sqrt{n})$ lower bound for $(2-\e)$-approximation}\label{sec:main}

In this section we prove Theorem~\ref{thm:main} and Theorem~\ref{thm:main-iid}.  We restate Theorem~\ref{thm:main} here for convenience of the reader.

\noindent{\em {\bf Theorem~\ref{thm:main}} 
Let $\e>0$ be a constant. Let $G=(V, E), |V|=n, |E|=m$ be an unweighted (multi)graph. Any algorithm that, given a single pass over a stream of edges of $G$ presented in random order, outputs a $(2-\e)$-approximation to the value of the maximum cut in $G$ with probability at least $99/100$ over its internal randomness must use $\tilde \Omega(\sqrt{n})$ space.
}

\noindent\begin{proof}
Consider running \ALG on the edges of the graph sampled from $\D$ that are presented in the canonical random ordering.  By Lemma~\ref{lm:perm} the total variation distance between the input graph in uniformly random order and the canonical random ordering is $O(\alpha \log(1/\alpha))$, so \ALG succeeds on $\D$ with probability at least $9/10$. By Yao's minmax principle there exists a deteministic algorithm {\bf ALG}' with at most the space complexity of \ALG that succeeds on $\D$ with probability at least $9/10$.

By Lemma~\ref{lm:reduction}, {\bf ALG}' can be used to construct a protocol for the \DBHP problem that gives advantage 
$\Omega(\alpha \e^2)$ over random guessing, where $\alpha<1$ is a parameter that remains to be set.

By Lemma~\ref{lm:lb-final} any protocol that uses at most $\gamma \sqrt{n}$ communication can not get advantage over random guessing larger than $O(\gamma+\alpha^{3/2})$.
We choose $\alpha=1/\log n$, which satisfies the preconditions of Lemma~\ref{lm:lb-final}. Substituting the value of $\alpha$, we get that one necessarily has $\gamma+(\log n)^{-3/2}= \Omega( \e^2/\log n)$, and hence $\gamma\geq C \e^2/\log n-\log^{-3/2} n\geq (C/2))\e^2/\log n$
for some constant $C>0$ for sufficiently large $n$. This implies that $\gamma=\Omega(1/\log n)$ for any constant $\e>0$, completing the proof.

\end{proof}

The rest of this section is devoted to proving Theorem~\ref{thm:main-iid}. We need one more ingredient for that. In particular, we now show that as long as the parameter $\alpha$ is sufficiently small, our distribution $\D$ is generating a sequence of edges of $G$ that is very close to a sequence of i.i.d.  samples in distribution. Intuitively, this is because while each of the $k$ phases of our input is distributed as $\G_{n, \alpha/n}$ as opposed to i.i.d., these distributions  are quite close in total variation distance. We make this claim precise below, obtaining a proof of Theorem~\ref{thm:main-iid}.

We will use the following lemmas, which state that the distribution of the stream of edges produced by our distributions $\D^Y$ and $\D^N$ is close in total variation to a stream of i.i.d. samples of edges of the complete bipartite graph and the complete graph respectively.

We first prove some auxiliary statements.
For the \YES case we start by establishing the following claim.
\begin{claim}\label{cl:iid-yes}
Let $G=(P, Q, E)$ be a complete bipartite graph, where $P\cup Q$ is a uniformly random partition of $[n]$.
Let $A=(A_1,A_2,\ldots, A_k)$ denote a sequence of $T$ i.i.d. samples of edges of $G$, where $T=T_1+T_2+\ldots+T_k$ is a sum of $k$ independent random variables distributed as $\text{Binomial}(|P|\cdot |Q|, \alpha/n)$ and $|A_i|=T_i, i=1,\ldots, k$. 

Let $B=(B_1,B_2,\ldots, B_k)$ denote a sequence of $T$ samples of edges of $G$, where for each $i=1,\ldots, k$ each $e\in E$ belongs to $B_i$ independently with probability $\alpha/n$. 

Then 
$$
||A-B||_{tvd}=O(k\alpha^2).
$$
\end{claim}

\begin{claim}\label{cl:pq}
Let $P, Q \subseteq V$ be a uniformly random bipartition of $V=[n]$. Then for any $\delta>0$ one has $\prob[||P||Q|-(n/2)^2|>\delta n]<e^{-\Omega(\delta)}$.
\end{claim}
\begin{proof}
Let $\Delta=|P|-n/2$. Then 
$|P||Q|=(n/2+\Delta)(n/2-\Delta)=(n/2)^2-\Delta^2$,
so 
$$
\prob[||P||Q|-(n/2)^2|>\delta n]=\prob[|\Delta|>\sqrt{\delta n}]<e^{-\Omega(\delta)}
$$
by Chernoff bounds.
\end{proof}

\begin{proofof}{Claim~\ref{cl:iid-yes}}
We first show that none of sets $A_i, i=1,\ldots, k$ contains repeated edges with probability $1-O(\alpha)$.

\begin{equation}\label{eq:aminusb}
\begin{split}
\prob[A_i\text{~contains a duplicate edge}]&\leq \sum_{e\in P\times Q}\prob[e\text{~appears more than once in~}A_i]\\
&\leq |P|\cdot |Q|\cdot (1-e^{-\lambda}-\lambda e^{-\lambda}),
\end{split}
\end{equation}
where $\lambda=T_i/m$ is the rate of arrival of an edge in $T_i$ samples in the i.i.d. setting. By Claim~\ref{cl:pq} one has 
$\prob[||P||Q|-(n/2)^2|>O(\log 1/\alpha) n]<\alpha^3$. Since $T_i$ is distributed as $\text{Binomial}(|P|\cdot |Q|, \alpha/n)$, we have $T_i\leq (\alpha/n)(n^2/4)+O(\log (1/\alpha)\alpha)$ with probability $1-\alpha^3$. Thus, we have 
$\lambda=T_i/m\leq 2(\alpha n/4)/(n^2/4)\leq 2\alpha/n$ with probability at least $1-O(\alpha^3)$. Using this in \eqref{eq:aminusb}, we get
\begin{equation*}
\begin{split}
\prob[A_i\text{~contains a duplicate edge}]&\leq |P|\cdot |Q|\cdot (1-e^{-\lambda}-\lambda e^{-\lambda})\\
&\leq |P|\cdot |Q|\cdot (1-e^{-\lambda}-\lambda e^{-\lambda})=|P|\cdot |Q|\cdot O(\lambda^2)=O(\alpha^2)\\
\end{split}
\end{equation*}
as required. By a union bound over all $i=1,\ldots, k$ we have that no $A_i$ contains a duplicate edge with probability at least $1-O(k\alpha^2)$. 
Conditional on not containing duplicate edges, $A_i$'s are uniformly random sets of edges of $G=(P, Q, E)$ of size $T_i$. Thus, $A$ has the same distribution as $B$ conditional on an event of probability at most $O(k\alpha^2)$, and hence $||A-B||_{tvd}=O(k\alpha^2)$ as required.
\end{proofof}

A similar claim holds for the \NO case:
\begin{claim}\label{cl:iid-no}
Let $G=(V, E)$ be a complete graph, and for each $i=1,\ldots, k$ let $E'_i\subseteq E$ be obtained by including every edge $e\in E={V\choose 2}$ independently with probability $1/2$. Let $A=(A_1,A_2,\ldots, A_k)$ denote a sequence of $T$ i.i.d. samples of edges of $G$ and $T=T_1+T_2+\ldots+T_k$ is a sum of $k$ independent random variables, where $T_i$ is distributed as $\text{Binomial}(|E'_i|, \alpha/n)$ and $|A_i|=T_i, i=1,\ldots, k$. 

Let $B=(B_1,B_2,\ldots, B_k)$ denote a sequence of $T$ samples of edges of $G$, where for each $i=1,\ldots, k$ each $e\in E'$ belongs to $B_i$ independently with probability $\alpha/n$. 

Then 
$$
||A-B||_{tvd}=O(k\alpha^2).
$$
\end{claim}
The proof of Claim~\ref{cl:iid-no} is essentially the same and is hence omitted.

We can now give 

\begin{proofof}{Theorem~\ref{thm:main-iid}}
  Suppose that \ALG yields a $(2-\e)$-approximation to maxcut with success probability at least $99/100$ on any fixed input if the stream contains $\ell\cdot n$ i.i.d. samples of the edges of the graph. We prove that \ALG must use $\tilde \Omega(\sqrt{n})$ space in two steps. First, we set up the parameters of the input distribution $\D$ so that \ALG must succeed with probability at least $9/10$ on $\D$. We then use Lemma~\ref{lm:lb-final} similarly to the proof of Theorem~\ref{thm:main} 

We choose the number of phases in our input as  $k=C\ell/(\alpha \e^2)$ for a constant $C>8$, and let $\alpha=\frac{1}{\ell^3 \log n}$ (note that this satisfies the condition $\alpha\in (\lba, 1)$). We will use the fact that 
\begin{equation}\label{eq:kalpha-small}
k\alpha^2=O(\frac{\ell}{\alpha \e^2}\cdot \alpha^2)= O(\ell\cdot \alpha)=o(1).
\end{equation}

We now show that with this setting of parameters the input stream contains a sequence of at least $\ell\cdot n$ samples of either a complete bipartite graph (\YES case) or a complete graph (\NO case), and the distribution of these samples is $o(1)$-close to i.i.d. in total variation distance (this claim will crucially rely our setting of parameters to ensure \eqref{eq:kalpha-small}). \footnote{We note that it is sufficient to prove a lower bound in the setting where the stream contains {\em at least $\ell\cdot n$} i.i.d. samples, since the algorithm can simply count the number of edges received and output the answer as soon as $\ell \cdot n$ have been received. This only increases the space complexity by an additive $O(\log n)$ term.} We consider the \YES and \NO cases separately.

\begin{description}
\item[\YES case.]  Let $P\cup Q=V$ denote the uniformly random bipartition used in the definition of $\D^Y$. Since $||P|-n/2|<O(\sqrt{n\log n})$ with probability $1-1/n^3$, say, by standard concentration inequalities, we have that each graph $G'_i$ for $i=1,\ldots,k$ contains at least $\alpha n/8$ edges with probability at least $1-1/n$ (we took a union bound over all $i=1,\ldots, k=O(n/\alpha)=O(n^{1.1})$). Thus, the union of $k=C\ell/(\alpha/\e^2)$ graphs generated by $\D^Y$ contains at least $\ell \cdot n$ edges with probability at least $1-1/n$.

Let $T^{Y}=T_1^Y+T_2^Y+\ldots+T_k^Y$ be a sum of $k$ independent random variables distributed as $\text{Binomial}(|P|\cdot |Q|, \alpha/n)$. 

By Claim~\ref{cl:iid-yes} the total variation distance between the stream of $T^Y$ i.i.d. samples of the complete bipartite graph $G=(P, Q, E)$ for the randomly chosen bipartition $P, Q$ and the stream of edges generated by $\D^Y$ is $O(k\alpha^2)$.  Thus, an algorithm \ALG that succeeds with probability at least $99/100$ on every input as long as the input stream contains at least $\ell\cdot n$ i.i.d. samples of the input graph must succeed at on $\D^Y$ with probability at least $99/100-O(k\alpha^2)\geq 9/10$ by \eqref{eq:kalpha-small}. 

\item[\NO case.]  We have that each graph $G'_i$ for $i=1,\ldots,k$ contains at least $\alpha n/8$ edges with probability at least $1-1/n$ (we took a union bound over all $i=1,\ldots, k\leq O(n^{1.1})$). Thus, the union of $k=C\ell/(\alpha/\e^2)$ graphs generated by $\D^N$ contains at least $\ell \cdot n$ edges with probability at least $1-1/n$.

Let $T^{N}=T^{N}_1+T^{N}_2+\ldots+T^{N}_k$ be a sum of $k$ independent random variables, where $T^{N}_i\sim \text{Binomial}(|E'_i|, \alpha/n)$ for a parameter $\alpha<1$.

By Claim~\ref{cl:iid-no} the total variation distance between the stream of $T^N$ i.i.d. samples of the complete graph $G=(V, E)$ and the stream of edges generated by $\D^N$ is $O(k\alpha^2)$. Thus, an algorithm \ALG that succeeds with probability at least $99/100$ on every input as long as the input stream contains at least $\ell\cdot n$ i.i.d. samples of the input graph must succeed at on $\D^N$ with probability at least $99/100-O(k\alpha^2)\geq 9/10$ by \eqref{eq:kalpha-small}.
\end{description}

 Thus, we have that \ALG must succeed with probability at least $9/10$ when input is drawn from distribution $\D$.  By Yao's minmax principle there exists a deterministic algorithm {\bf ALG}' with space complexity bounded by that of \ALG that succeeds with probability at least $9/10$ on $\D$. Now by Lemma~\ref{lm:reduction} there exists a protocol for the \DBHP problem that gives advantage $\Omega(1/k)=\Omega(\frac1{\ell}\alpha \e^2)$ over random guessing. 

However, for any $\gamma>\lba$ by Lemma~\ref{lm:lb-final} any protocol that uses at most $\gamma \sqrt{n}$ communication can not get advantage over random guessing larger than 
\begin{equation}\label{eq:23hfas}
O(\gamma+\alpha^{3/2}).
\end{equation}
Substituting the value of $\alpha$ into ~\eqref{eq:23hfas}, we get that one necessarily has
$$
\gamma+\left(\ell^3 \log n\right)^{-3/2}= \Omega(\frac{\e^2}{\ell^4 \log n}),
$$
and hence 
$$
\gamma\geq C \frac{\e^2}{\ell^4 \log n}-\ell^{-9/2}/(\log n)^{3/2}
$$
for some constant $C>0$. This implies that $\gamma=\Omega(\frac{\e^2}{\ell^4 \log n})$ for any constant $\e>0$, completing the proof.
\end{proofof}

\newpage
\pdfbookmark[1]{\refname}{My\refname} 
%\bibliographystyle{plain}
%\bibliography{sparsification}

\begin{appendix}
%\input{app.tex}
%\input{appb.tex}
%!TEX root = ./factor2.tex
\section{Omitted Proofs}\label{app:reduction}

In this section we give the proofs of Claim~\ref{cl:1} and Lemma~\ref{lm:single-sample}. 

\noindent{\em {\bf Claim~\ref{cl:1}}
Let $X, Y$ be two random variables. Let $W$ be independent of $(X, Y)$. Then for any function $f$ one has 
$||f(X, W)-f(Y, W)||_{tvd}\leq ||X-Y||_{tvd}$.
}

\noindent\begin{proof}
First, one has $||(X, W)-(Y, W)||_{tvd}=||X-Y||_{tvd}$  since $W$ is independent of $X$ and of $Y$.  The claim now follows since $||f(A)-f(B)||_{tvd}\leq ||A-B||_{tvd}$ for any $A, B$, we demonstrated by the following calculation. Suppose that $f: \Omega\to \Omega'$. Then
\begin{equation*}
\begin{split}
||f(A)-f(B)||_{tvd}&=\frac1{2}\sum_{\omega\in \Omega} |p_A(f^{-1}(\omega))-p_B(f^{-1}(\omega))|\\
&=\frac1{2}\sum_{\omega'\in \Omega'} \left|\sum_{\omega\in \Omega: f(\omega)=\omega'} p_A(\omega)-p_B(\omega)\right|\\
&\leq \frac1{2}\sum_{\omega'\in \Omega'} \sum_{\omega\in \Omega: f(\omega)=\omega'} \left|p_A(\omega)-p_B(\omega)\right|\\
&\leq \frac1{2}\sum_{\omega\in \Omega} \left|p_A(\omega)-p_B(\omega)\right|,\\
\end{split}
\end{equation*}
where we used the fact that $|a+b|\leq |a|+|b|$ for any $a, b\in \mathbb{R}$ to go from line 2 to line 3.
\end{proof}
\vspace{0.2in}
\noindent {\em {\bf Lemma~\ref{lm:single-sample}}
Let $X, Y$ be distributions on a finite probability space $\Omega$.  Suppose that with probability $1/2$ one is given a sample $\omega$ from $X$ (\YES case) and with probability $1/2$ a sample from $Y$ (\NO case). Then outputting \YES if $p_X(\omega)>p_Y(\omega)$ and \NO otherwise distinguishes between the two cases with advantage over random guessing at least $||X-Y||_{tvd}/2$.}

\noindent\begin{proof}
Recall that 
$$
||X-Y||_{tvd}=\max_{\Omega'\subseteq \Omega} (p_X(\Omega')-p_Y(\Omega'))=\frac1{2}\sum_{\omega\in \Omega} |p_X(\omega)-p_Y(\omega)|.
$$
Let $\Omega^*$ be the set that achieves the optimum, i.e.
$$
\Omega^*=\{\omega\in \Omega : p_X(\omega)>p_Y(\omega)\}.
$$

The probability of error equals

\begin{equation*}
\begin{split}
&\prob[\text{answer is~}\YES]\sum_{\omega\in \Omega^*} p_Y(\omega)+\prob[\text{~answer is~}\NO]\sum_{\omega\in \Omega\setminus \Omega^*} p_X(\omega)\\
&=\frac1{2}\sum_{\omega\in \Omega^*} (p_X(\omega)-(p_X(\omega)-p_Y(\omega)))+\frac1{2}\sum_{\omega\in \Omega\setminus \Omega^*} p_X(\omega)\\
&=\frac1{2}-\frac1{2}||X-Y||_{tvd}
\end{split}
\end{equation*}
\end{proof}

\begin{proofof}{Lemma~\ref{lm:tvd-cond}}
\begin{equation*}
\begin{split}
||(X, Y^1)-(X, Y^2)||_{tvd}&=\frac1{2}\sum_{x\in \Omega_1, y\in \Omega_2} |p_{(X, Y^1)}-p_{(X, Y^2)}|\\
&=\frac1{2}\sum_{x\in \Omega_1, y\in \Omega_2} |p_{(X, Y^1)}-p_{(X, Y^2)}|\\
&=\frac1{2}\sum_{x\in \Omega_1, y\in \Omega_2} |p_{X}(x) p_{Y^1_x}(y)-p_{X}(x) p_{Y^2_x}(y)|\\
&=\frac1{2}\sum_{x\in \Omega_1}p_{X}(x)\sum_{y\in \Omega_2} |p_{Y^1_x}(y)-p_{Y^2_x}(y)|\\
&=\sum_{x\in \Omega_1}p_{X}(x)||Y^1_x-Y^2_x||_{tvd}\\
&=\expect_X[||Y^1_X-Y^2_X||_{tvd}]\\
\end{split}
\end{equation*}
\end{proofof}

\end{appendix}

\end{document}